\documentclass[sigconf,table]{acmart}

\usepackage{booktabs}
\usepackage{enumitem}

\newcolumntype{L}[1]{>{\raggedright\let\newline\\\arraybackslash\hspace{0pt}}m{#1}}
\newcolumntype{C}[1]{>{\centering\let\newline\\\arraybackslash\hspace{0pt}}m{#1}}
\newcolumntype{R}[1]{>{\raggedleft\let\newline\\\arraybackslash\hspace{0pt}}m{#1}}

\newcommand\independent{\protect\mathpalette{\protect\independenT}{\perp}}
\def\independenT#1#2{\mathrel{\rlap{$#1#2$}\mkern2mu{#1#2}}}

\setcopyright{none}

\acmConference[Working paper]{Working paper}{\today}{Stanford University}
\acmYear{2017}
\copyrightyear{2017}

\acmPrice{}

\begin{document}

\title{Simple Rules for Complex Decisions}

\author{Jongbin Jung}
\affiliation{%
    \institution{Stanford University}
}
\email{jongbin@stanford.edu}

\author{Connor Concannon}
\affiliation{%
    \institution{John Jay College of Criminal Justice}
}
\email{cconcannon@jjay.cuny.edu}

\author{Ravi Shroff}
\affiliation{%
    \institution{New York University}
}
\email{ravi.shroff@nyu.edu}

\author{Sharad Goel}
\affiliation{%
    \institution{Stanford University}
}
\email{scgoel@stanford.edu}

\author{Daniel G. Goldstein}
\affiliation{%
    \institution{Microsoft Research}
}
\email{dgg@microsoft.com}

\begin{abstract}
From doctors diagnosing patients to judges setting bail, experts often base
their decisions on experience and intuition rather than on statistical models.
While understandable, relying on intuition over models has often been found to result
in inferior outcomes.
Here we present a new method---\emph{select-regress-and-round}---for constructing simple rules that perform well for complex decisions.
These rules take the form of a weighted checklist, can be
applied mentally, and nonetheless rival the performance of modern machine learning algorithms.
Our method for creating these rules is itself simple,
and can be carried out by practitioners with basic statistics knowledge.
We demonstrate this technique with a detailed case study of judicial decisions
to release or detain defendants while they await trial.
In this application, as in many policy settings,
the effects of proposed decision rules cannot be directly observed from historical data:
if a rule recommends releasing a defendant that the judge in reality detained,
we do not observe what would have happened under the proposed action.
We address this key counterfactual estimation problem by drawing on tools from causal inference.
We find that simple rules significantly outperform judges and are on par with decisions derived from random forests trained on all available features.
Generalizing to 22 varied decision-making domains, we find this basic result replicates.
We conclude with an analytical framework that helps explain why
these simple decision rules perform as well as they do.
\end{abstract}

\maketitle
\renewcommand{\shortauthors}{Jung et al.}

\section{Introduction}
In decision-making scenarios, experts often choose a course of action
based on experience and intuition rather than on statistical
analysis~\citep{Gigerenzer2011}.
This includes doctors classifying patients based on their
symptoms~\citep{McDonald1996}, judges setting bail amounts~\citep{Dhami2003}
and making parole decisions~\citep{danziger_2011}, and managers determining
which customers to target~\citep{wubben2008instant}.
A large body of work shows that intuitive judgments are generally inferior
to those based on statistical
models~\citep{dawes1979robust,dawes1989clinical,tetlock_2005,kleinberg_2015,kleinberg_2017}.
However, decision makers have consistently eschewed formal decision models in
part because it has been difficult to create, understand, and apply them.

Here we present a simple method for constructing simple decision
rules that often perform on par with traditional machine learning algorithms.
Our \emph{select-regress-and-round} strategy results in
rules that are fast, frugal, and clear:
fast in that decisions can be made quickly
in one's mind, without the aid of a computing device; frugal in that they
require only limited information to reach a decision; and clear in that they
expose the grounds on which classifications are made.
Decision rules satisfying these criteria have many benefits.
For instance, rules that can be applied quickly
and mentally are likely to be adopted and used persistently. In medicine,
frugal rules require fewer tests, which saves time, money, and, in the case of
triage situations, lives~\citep{marewski2012heuristic}. The clarity of simple
rules engenders trust from users, providing insight into how systems work and
exposing where models may be improved~\citep{gleicher2016,sull_2015}. Clarity
can even become a legal requirement when society demands to know how
algorithmic decisions are being made~\citep{goodman2016eu,corbett-davies_2017}.

Our results add to a growing literature on \emph{interpretable machine
learning}~\cite{kim2014,kim2015,ustun_2016,letham_2015,lakkaraju_2016}.
Several methods recently have  been introduced to construct the kind of simple decision rules
we discuss here, including supersparse linear integer models
(SLIM)~\citep{ustun_2016,ustun_2017}, Bayesian rule lists~\citep{letham_2015}, and
interpretable decision sets~\citep{lakkaraju_2016}.
These methods all produce rules that are easy to interpret and to apply.
One important difference between our approach and past techniques is that our rules are also easy to create.

To illustrate our method, we begin with a case study of judicial decisions for pretrial release. We show that simple rules substantially improve upon the efficiency and equity of unaided decisions.
In particular, we estimate that judges can detain half as many defendants without appreciably increasing the number that fail to appear at their court dates.
Our simple rules perform as well as
a black-box, random forest model trained on all available data.
(We note that~\citet{kleinberg_2017} recently and independently proposed using random forests to assist judicial decisions, but they do not consider simple rules.)
We further evaluate the efficacy of our method on 22
datasets from the UCI ML repository and show that
in many cases simple rules are competitive with state-of-the-art machine learning algorithms.
We conclude with an
analytical framework that helps explain why
simple decision rules often perform well.

\section{Illustration: bail decisions}
\label{sec:bail}

As an initial example of how to create simple rules that make accurate and transparent
decisions, we turn to the domain of pretrial release determinations.
In the United States,
a defendant is typically arraigned shortly after arrest in a court appearance
where he is provided with written notice of the charges alleged by the
prosecutor. At this time, a judge must decide whether the defendant, while he
awaits trial, should be \emph{released on his own recognizance} (RoR), or
alternatively, subject to monetary bail.
In practice, if the judge rules that bail be set, defendants often await trial
in jail since many of them do not have the financial resources to post bail.
Moreover, when defendants are able to post bail, they often do so by contracting
with a bail bondsman and in turn incur hefty fees.
The judge, however, has a legal obligation
to consider taking measures necessary to secure the defendant's
appearance at required court proceedings.
Pretrial release decisions must thus balance flight risk against the high burden
that bail requirements place on defendants.
In many jurisdictions judges may also consider a defendant's threat to public safety, but that is not a legally relevant factor for the specific jurisdiction we analyze below.

A key statistical challenge in this setting is that
one cannot, with historical data alone, directly observe the effects of hypothetical decision rules.
For example, if a proposed policy recommends releasing some defendants who in reality were detained by the judge, one does not observe what would have happened had the rule been followed.
This counterfactual estimation problem---also known as offline policy evaluation~\citep{dudik_2011}---is common in many domains.
We address it here by adapting tools from causal inference to the policy setting,
including the method of \citet{rosenbaum_1983a}
for assessing the sensitivity of estimated causal effects to an unobserved
binary covariate.

Our analysis is based on 165,000 adult cases involving nonviolent offenses
charged by a large urban prosecutor's office and arraigned in criminal court
between 2010 and 2015.
This set was obtained by starting with a random sample of 200,000 cases provided to us by the prosecutor's office, and then restricting to those cases involving
nonviolent offenses and for which the records were complete and accurate.
Our initial sample of 200,000 cases does not include instances where
defendants accepted a plea deal at arraignment, obviating
the need for a pretrial release decision.
For each case, we have a rich set of attributes: 49 features describe
characteristics of the current charges (\emph{e.g.}, theft, gun-related), and
15 describe characteristics of the defendant (\emph{e.g.}, gender, age, prior
arrests). We also observe whether the defendant was RoR'd, and whether he
failed to appear (FTA) at any of his subsequent court dates. We note that even
if bail is set, a defendant may still fail to appear since he could post bail
and then skip his court date. Overall, 69\% of defendants are RoR'd, and 15\%
of RoR'd defendants fail to appear. Of the remaining 31\% of defendants for whom bail is
set, 45\% are eventually released and 9\% fail to appear. As a result, the overall FTA
rate is 13\%.

In our analysis below, we randomly divide the full set of 165,000 cases into three
approximately equal subsets; we use the first fold to construct decision rules (both simple and complex),
and the second and third to evaluate these rules, as described next.

\subsection{Rule construction}

We start by constructing traditional (but complex) decision rules for balancing
flight risk with the burdens of bail. These rules serve as a benchmark
for evaluating the simple rules we create below.
On the first fold of the data, we restrict to cases in which the judge RoR'd
the defendant, and then train a random forest model to estimate the likelihood
an individual fails to appear at any of his subsequent court dates.
Random forests are considered to be one of the best off-the-shelf classification
algorithms~\citep{fernandez2014,kleinberg_2017}, and
we fit the model on all available information about the case and the
defendant, excluding race.\footnote{%
We use the \texttt{randomForest} package in \texttt{R}, fit with 1,000 trees.
We exclude race from the presented results due to legal and policy concerns with basing decisions on protected attributes~\cite{corbett-davies_2017}. We note, however, that including race does not significantly affect performance.}
The fitted model lets us compute risk scores (\emph{i.e.}, estimated flight risk
if RoR'd) for any defendant. These risk scores can in turn be converted to
a binary decision rule by selecting a threshold for releasing individuals. One
might, for example, RoR a defendant if and only if his flight risk is below
20\%.

We now construct a family of simple rules for making release decisions.
We begin by fitting a logistic regression model that estimates a defendant's
flight risk as a function of his age and prior history of failing to appear.
These two factors are well understood to be highly predictive in this context,
but we later show how such features can be selected in a principled fashion
without domain expertise.
Specifically, we fit the following model:
\begin{align*}
  & \Pr(Y_i = 1) = \text{logit}^{-1} \Big(\beta_0  + \\
  & \hspace{5mm} \beta^{\text{priors}}_{1} H_i^{1} +
    \beta^{\text{priors}}_{2} H_i^{2} +
    \beta^{\text{priors}}_{3} H_i^{3} +
    \beta^{\text{priors}}_{4+} H_i^{4+} + \\
  & \hspace{5mm} \beta^{\text{age}}_{18-20} A_i^{18-20} +
      \beta^{\text{age}}_{21-25} A_i^{21-25} + \cdots +
      \beta^{\text{age}}_{46-50} A_i^{46-50}
    \Big),
\end{align*}
where $Y_i \in \{0, 1\}$ indicates whether the $i$-th defendant failed to appear;
$H_i^* \in \{0, 1\}$ indicates the defendant's number of past failures to appear
(exactly one, two, three, or at least four);
and $A_i^* \in \{0, 1\}$ indicates the binned age of the defendant
(18--20, 21--25, 26--30, 31--35, 36--40, 41--45, or 46--50).
For identifiability, indicator variables for zero past FTAs and age 51-and-older
are omitted. As before, this model is fit on the subset of cases in the first
fold of data for which the judge released the defendant.
Next, we rescale the age and prior FTA coefficients so that they lie in the
interval $[-10, 10]$; specifically we multiply each coefficient by the constant
\begin{equation*}
\frac{10}{\max \left(|\beta^{\text{priors}}_{1}|, \dots, |\beta^{\text{priors}}_{4+}|, |\beta^{\text{age}}_{18-20}|, \dots, |\beta^{\text{age}}_{46-50}|\right)}.
\end{equation*}
Finally, we round the rescaled coefficients to the nearest integer.

\begin{table}
  \caption{%
    A defendant's flight risk is obtained by adding the scores
    for age and prior failure to appear (FTA).
  }
  \begin{tabular*}{8cm}{@{\extracolsep{\fill}}lrclr}
    \toprule
    Feature & Score & & Feature & Score\\
    \midrule
    $18 \leq $ age $< 21$ & 8 & &
    no prior FTAs & 0 \\
    $21 \leq $ age $< 26$ & 6 & &
    1 prior FTA & 6 \\
    $26 \leq $ age $< 31$ & 4 & &
    2 prior FTAs & 8 \\
    $31 \leq $ age $< 51$ & 2 & &
    3 prior FTAs & 9 \\
    $51 \le$ age & 0 & &
    4 or more prior FTAs & 10 \\
    \bottomrule
  \end{tabular*}
\label{tab:heuristic}
  \vspace{-4mm}
\end{table}

Table~\ref{tab:heuristic} shows the result of this procedure.
For any defendant, a risk score can be computed by summing the relevant terms in the table.
Unsurprisingly, past FTAs are indeed strong predictors of future failure to appear;
an individual's risk also declines with age, in line with conventional wisdom.
These risk scores can be converted to a binary decision rule by selecting
a threshold for releasing individuals.
For example, one might RoR a defendant if and only if his risk score is below
10.5.
A graphical representation of that rule is shown in Figure~\ref{fig:policy-grid}.

\begin{figure}
  \includegraphics[width=5cm]{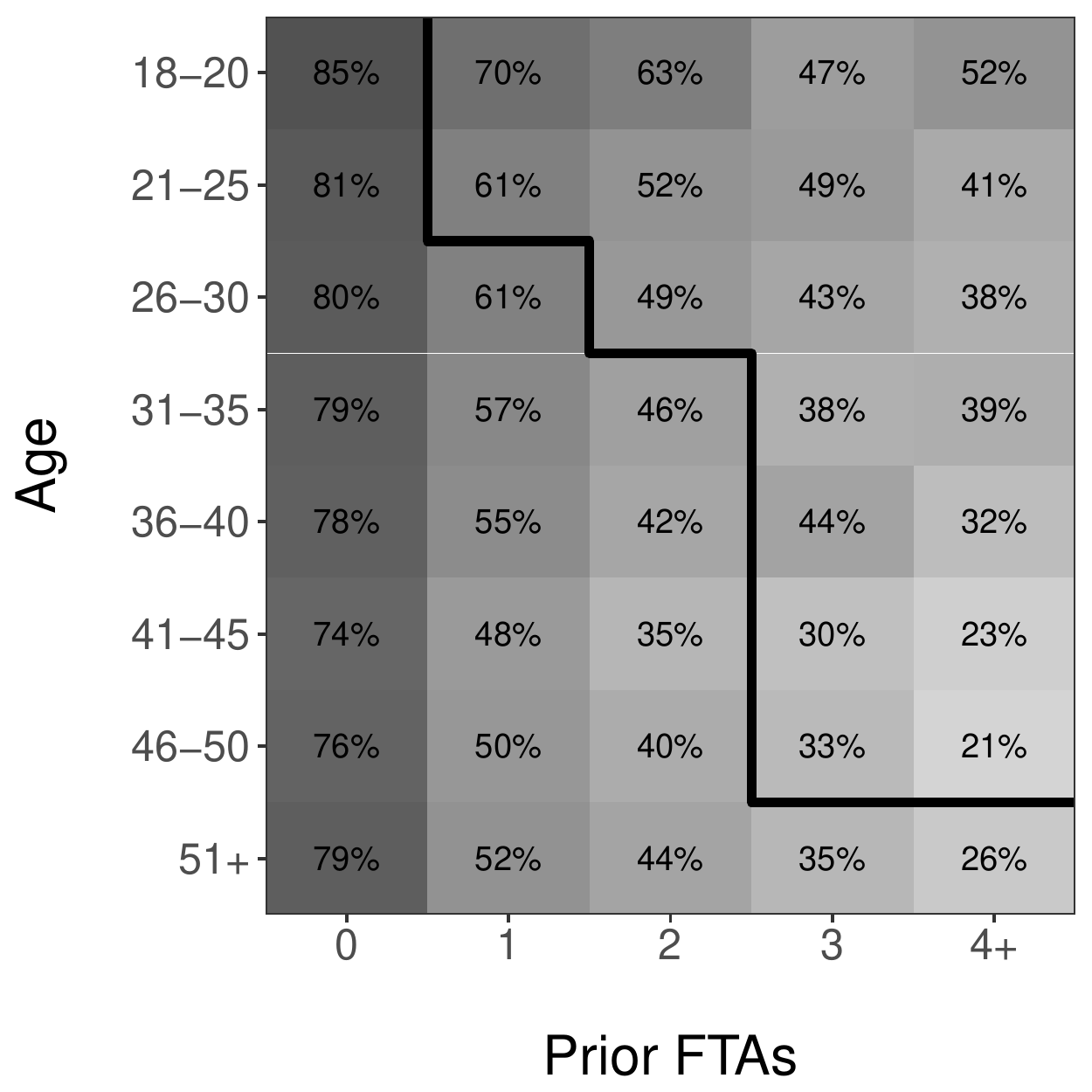}
  \caption{%
    Graphical representation of a simple rule for release decisions, based on
    setting a release threshold of 10.5 on the risk scores described in
    Table~\ref{tab:heuristic}.
    Groups to the left of the black line are those that would be released
    under the rule; for comparison, the shading and numbers show the proportion
    of defendants that are currently RoR'd in each group.
  }
\label{fig:policy-grid}
\end{figure}

\subsection{Policy evaluation}
\label{sec:evaluation}

There are two key considerations in evaluating a decision rule for pretrial
release: (1) the proportion of defendants who are released under the rule; and
(2) the resulting proportion who fail to appear at their court proceedings. It
is straightforward to estimate the former, since one need only apply the rule
to historical data to see what actions would have been
recommended.\footnote{In theory, implementing a decision rule could alter the
equilibrium distribution of defendants. We do not consider such possible
effects, and assume the distribution of defendants is not affected by the rule
itself.} For example, if defendants are released if and only if their risk
score is below 10.5, 84\% would be RoR'd;
under this rule, bail would be required of only half as many defendants
relative to the status quo.
Forecasting the proportion who would fail to appear, however, is
generally much more difficult. The key problem is that for any particular
defendant, we only observe the outcome (\emph{i.e.}, whether or not he failed
to appear) conditional on the action the judge ultimately decided to take
(\emph{i.e.}, RoR or bail). Since the action taken by the judge may differ from
that prescribed by the decision rule, we do not always observe what would have
happened under the rule. This problem of \emph{offline policy
evaluation}~\citep{dudik_2011}
is a specific instance of the fundamental problem of causal
inference.

To rigorously describe the estimation problem and our approach,
we first introduce some notation.
We denote the observed set of cases by $\Omega = \{(x_i,a_i,r_i)\}$, where $x_i$
is a case, $a_i \in  \{\text{RoR}, \text{bail}\}$ is the action taken by the
judge, and $r_i \in \{0,1\}$ indicates whether the defendant failed to appear at
his scheduled court date.
We write $r_i(\text{RoR})$ and $r_i(\text{bail})$ to mean the \emph{potential
outcomes}, what would have happened under the two possible judicial actions.
For any policy $\pi$, our goal is to estimate the FTA rate under the policy:
\begin{equation*}
  V^{\pi} = \frac{1}{|\Omega|} \sum_i r_i(\pi(x_i))
\end{equation*}
where $\pi(x)$ denotes the action prescribed under the rule.
The key statistical challenge is that only one of the two potential outcomes, $r_i = r_i(a_i)$, is observed.
We note that policy evaluation is a generalization of estimating average
treatment effects.
Namely, the average treatment effect can be expressed as
$V^{\pi_{\text{RoR}}} - V^{\pi_{\text{bail}}}$,
where $\pi_{\text{RoR}}$ is the policy under which everyone is released and
$\pi_{\text{bail}}$ is defined analogously.

Here we take a straightforward and popular statistical approach to estimating $V^{\pi}$: response
surface modeling~\citep{hill2012}.
With response surface modeling, the idea is to use a standard prediction model
(\emph{e.g.}, logistic regression or random forest) to estimate the effect on
each defendant of each potential judicial action.
The model estimates of these potential outcomes are denoted by
$\hat{r}_i(t)$, for $t \in \{\text{RoR}, \text{bail}\}$.
Our estimate of $V^{\pi}$ is then given by
\begin{equation*}
\hat{V}^{\pi} = \frac{1}{|\Omega|} \sum_i \big [
  r_i \mathbf{I}(\pi(x_i) = a_i) + \hat{r}_i (\pi(x_i))
  \mathbf{I}(\pi(x_i) \neq a_i)
\big ]
\end{equation*}
where $\mathbf{I}(\cdot)$ is an indicator function evaluating to 1 if its
argument is true and to 0 otherwise.
If the prescribed action is in fact taken by the judge, then $r_i
= r_i(\pi(x_i))$ is directly observed and can be used;
otherwise we approximate the potential
outcome with $\hat{r}_i(\pi(x_i))$.
Table~\ref{table:potential-outcomes} illustrates this method for a hypothetical
example.

Response surface modeling implicitly assumes that a judge's action is \emph{ignorable}
given the observed covariates (\emph{i.e}., that conditional on the observed covariates,
those who are RoR'd are similar to those who are not).
Formally, ignorability means that
\begin{equation*}
  (r(\text{RoR}), r(\text{bail})) \independent a \ \big | \, x.
\end{equation*}
This ignorability assumption is unavoidable, and is similarly required for methods based on propensity scores~\cite{rosenbaum_1983b,rosenbaum_1984,cassel_1976,robins_1994,robins_1995,kang_2007,dudik_2011}.
We examine this assumption in detail in Section~\ref{ssec:sensitivity},
and find that our conclusions are robust to unobserved
heterogeneity.

\begin{table}
  \caption{%
    For each defendant, $\hat{Y}_{\mathrm{RoR}}$ and $\hat{Y}_{\mathrm{bail}}$
    are model-based estimates of the likelihood of FTA under each potential
    action.  In cases where the observed action equals the proposed action, the
    observed outcome (FTA or not) is used to estimate the policy's effect; 		     otherwise, the model-based estimates are used.  The gray shading
    indicates which values are used in each instance.  The overall FTA rate
    under the policy is estimated by averaging the shaded values over all cases.
  }
\label{table:potential-outcomes}
  \begin{tabular}{C{1.3cm}C{1.3cm}C{1.3cm}C{1.3cm}C{1.3cm}}
    \toprule
    Proposed action & Observed action & Observed outcome
    & $\hat{Y}_{\mathrm{RoR}}$ & $\hat{Y}_{\mathrm{bail}}$\\
    \midrule
    RoR & RoR & \cellcolor{gray!25}0 & 20\% & 10\% \\
    Bail & Bail & \cellcolor{gray!25}1 & 80\% & 30\% \\
    Bail & RoR & 1 & 90\% & \cellcolor{gray!25}70\% \\
    RoR & Bail & 0 & \cellcolor{gray!25}30\% &  25\%\\
    RoR & RoR & \cellcolor{gray!25}0 & 20\% & 15\% \\
    \bottomrule
  \end{tabular}
\end{table}

To carry out this approach, we derive estimates $\hat{r}_i(t)$ via an
$L^1$-regularized logistic regression (lasso) model trained on the second fold
of our data.
For each individual, the model estimates his likelihood of FTA given all the observed features and the action taken by the judge.
In contrast to the rule construction described above, this time we train the
model on all cases (not just those for which the judge RoR'd the defendant) and
include as a predictor the judge's action (RoR or bail); we also include the defendant's race.\footnote{%
Although it is legally problematic to use race when \emph{making} decisions, its use is acceptable---and indeed often required---when \emph{evaluating} decisions.
The model was fit with the \texttt{glmnet} package in \texttt{R}.
The \texttt{cv.glmnet} method was used to determine the best value for the
regularization parameter $\lambda$ with 10-fold cross-validation and 1,000 values of $\lambda$.
The model includes all pairwise interactions between the judge's decision and defendant's features.
We opt for lasso instead of random forest for this prediction task because
the latter, while very good for classification, is known to suffer from poor
calibration~\citep{niculescu_2005}, which can in turn yield biased estimates
of a policy's effects.
}
Then, on the third fold of the data,
we use the observed and model-estimated outcomes to
approximate the overall FTA rate for any decision rule.

\begin{figure}
  \includegraphics[width=5cm]{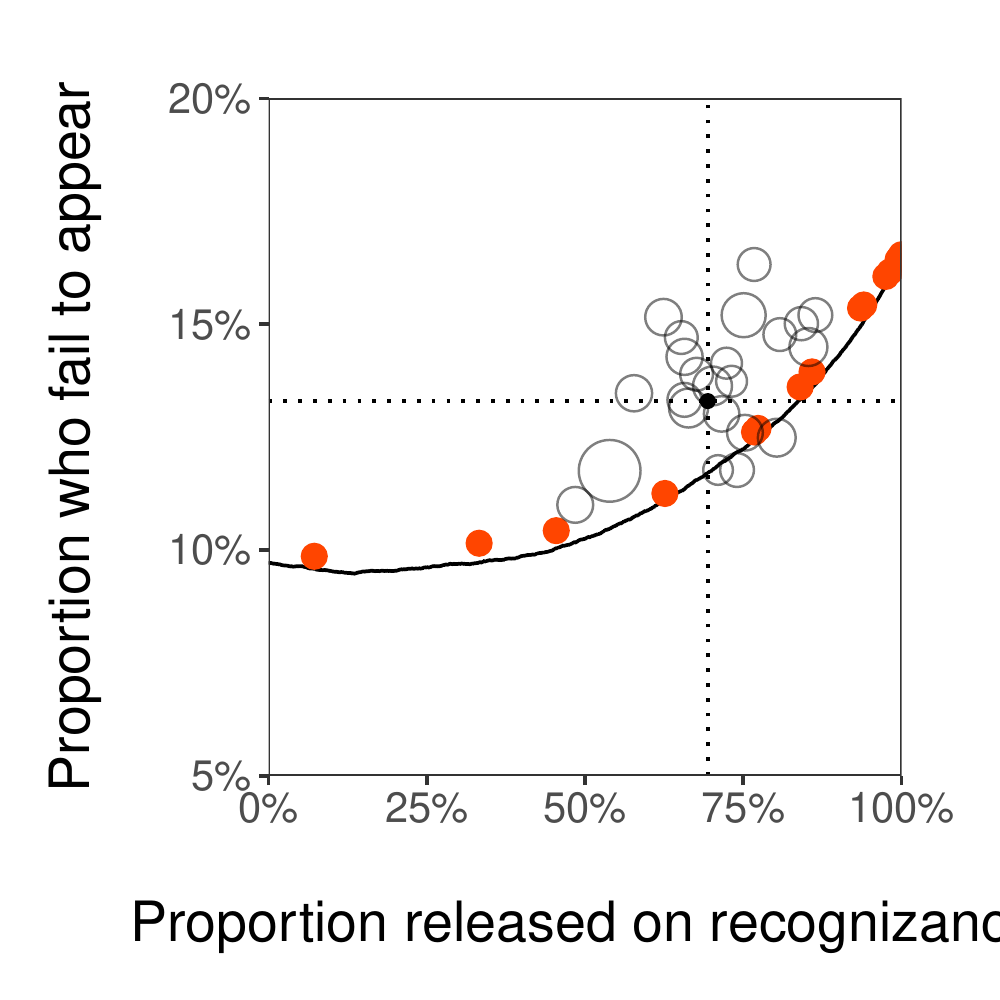}
  \caption{%
    Each point on the solid line corresponds to decision rules derived from
    a random forest risk model with varying thresholds for release.  The red
    points correspond to the simple risk score in Table~\ref{tab:heuristic} for
    all possible release thresholds.  The simple rules perform nearly
    identically to the complex models.  The open circles show the observed RoR
    and FTA rates for each judge in our data who presided over at least 1,000
    cases, sized in proportion to their case load.  In nearly every instance,
    the statistical decision rules outperform the human decision-maker.
  }
\label{fig:evaluation}
  \vspace{-5mm}
\end{figure}

Figure~\ref{fig:evaluation} shows estimated RoR and FTA rates for a variety of
pretrial release rules.  Points on the solid line correspond to rules
constructed via the random forest model described above for various decision
thresholds.  The red points correspond to rules based on the simple scoring
procedure in Table~\ref{tab:heuristic}, again corresponding to various decision
thresholds.  For each rule, the horizontal axis shows the estimated proportion
of defendants ROR'd under the rule, and the vertical axis shows the estimated
proportion of defendants who would fail to appear at their court dates.
The solid black dot shows the status quo: 69\% of defendants RoR'd and a 13\%
FTA rate. Finally, the open circles show the observed RoR and FTA rates for each
of the 23 judges in our data who have presided over at least 1,000 cases, sized
in proportion to their case load.

The plot illustrates three key points. First, simple rules that consider only
two features---age and prior FTAs---perform nearly identically to a random
forest that incorporates 64 features.
Second, the statistically informed policies in the lower right quadrant all
achieve higher rates of RoR and, simultaneously, lower rates of FTA than the
status quo. In particular, by releasing defendants if and only if their risk
score is below 10.5, we expect to release 84\% of defendants while achieving an
FTA rate of 14\%. Relative to the existing policy, following this rule
would not appreciably increase the overall FTA rate---it would increase just 0.3 percentage points, from 13.3\% to 13.6\%---but
only half as many defendants would be required to pay bail.
Finally, for nearly every judge, there is a statistical
decision rule that simultaneously yields both a higher rate of release and a
lower rate of FTA than the judge currently achieves. The statistical decision
rules consistently outperform the human decision-makers.

Why do these statistical decision rules outperform the experts?
Figure~\ref{fig:policy-grid} sheds light on this phenomenon. Each cell in
the plot corresponds to defendants binned by their age and prior number of
FTAs. Under a rule that releases defendants if and only if their risk score is
below 10.5, one would release everyone to the left of the solid black line, and
set bail for everyone to the right of the line. The number in each cell shows
the proportion of defendants in each bin who are currently released, and the
cell shading graphically indicates this proportion. Aside from the lowest risk
defendants, who have no prior FTAs, the likelihood of being released does not
correlate strongly with estimated flight risk. For example, the high risk group
of young defendants with four or more prior FTAs is released at about the same
50\% rate as the low risk group of older defendants with one prior FTA\@.
This low correlation between flight risk and release decision is in part
attributable to extreme differences in release rates across judges, with some
releasing more than 90\% of defendants and others releasing just 50\%.\footnote{Defendants are not perfectly randomly assigned to judges for arraignment, but
in practice judges see a similar distribution of defendants.}
Whereas defendants experience dramatically different outcomes based on the
judge they happened to appear in front of, statistical decision rules improve
efficiency in part by ensuring consistency.

\subsection{Sensitivity to unobserved heterogeneity}
\label{ssec:sensitivity}

As noted above, our estimation strategy assumes that
the judicial action taken is ignorable given the observed covariates.
Under this ignorability assumption, one can accurately estimate the potential outcomes.
Judges, however, might base their decisions in part on information
that is not recorded in the data, which could in turn bias our estimates.
For example, a judge, upon meeting a defendant, might surmise that his flight
risk is higher than one would expect based on the recorded covariates alone, and
may accordingly require the defendant to post bail.
In this case, since our estimates are based only on the recorded data, we may
underestimate the defendant's counterfactual likelihood of failing to appear if
released.

We take two approaches to gauge the robustness of our results to such hidden
heterogeneity.
First, on each subset of cases handled by a single judge, we use response
surface modeling to estimate $V^\pi$. Each judge has idiosyncratic criteria for
releasing defendants, as evidenced by the dramatically different release rates
across judges; accordingly, the types and proportion of cases for which the policy
$\pi$ coincides with the observed action differ from judge to judge.
This variation allows us to assess the sensitivity of our estimates to the
observed actions $\{a_i\}$.
In particular, if unobserved heterogeneity were significant, we would expect our estimates to systematically vary depending on the proportion
of observed judicial actions that agree with the policy $\pi$.
Figure~\ref{fig:sensitivity-judge} shows the results of this analysis for the simple
decision rule described in Figure~\ref{fig:policy-grid}, where each point
corresponds to a judge.
We find that the FTA rate of the decision rule is consistently estimated to be
approximately 12--14\%. Moreover, some judges act in concordance with the
decision rule in nearly 80\% of cases;
for this subset of judges, where our estimates are largely based on directly
observed outcomes, we again find FTA is estimated at around 12--14\%.

\begin{figure}
  \includegraphics[width=6cm]{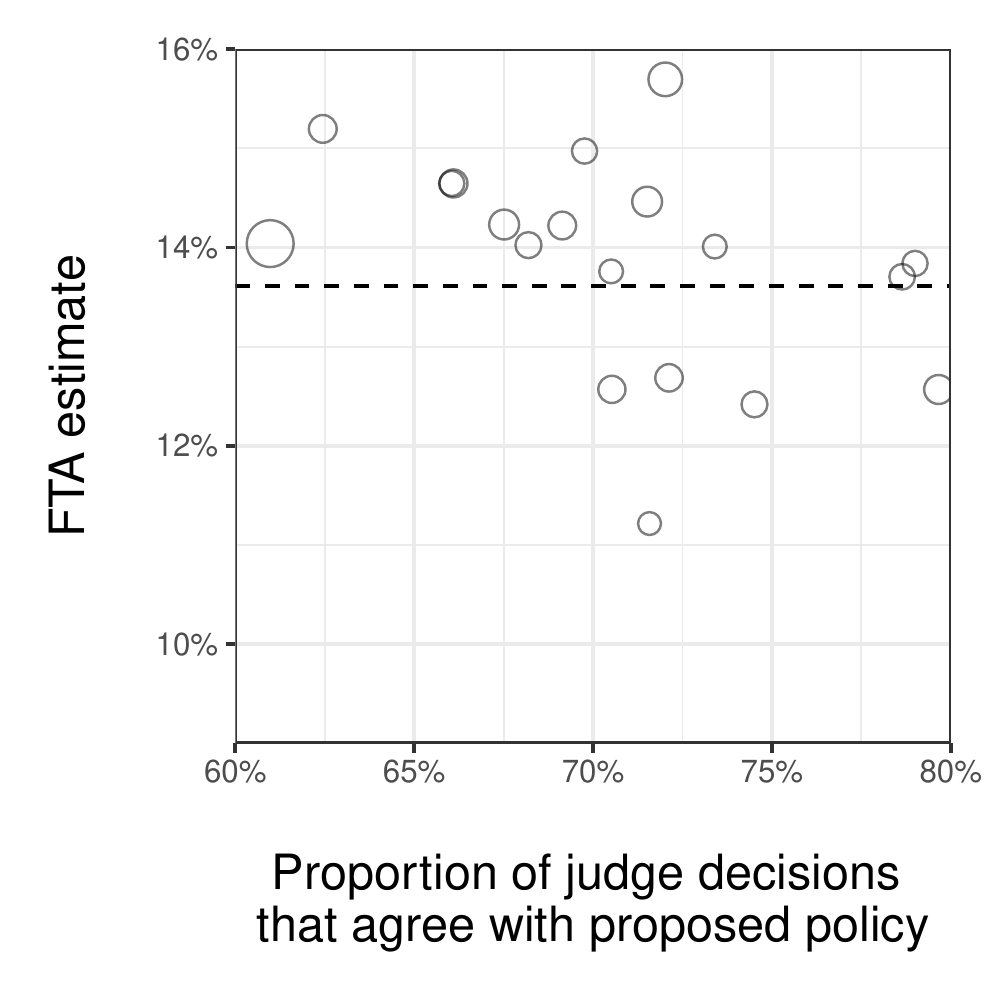}
  \caption{%
    For the simple decision rule illustrated in Figure~\ref{fig:policy-grid},
    FTA rate is estimated by separately applying response surface modeling to each
    judge's cases, where each point corresponds to a judge; the dashed
    horizontal line indicates the FTA rate of the decision rule estimated on
    the full set of cases.
    Though judges have different criteria for releasing defendants---and the corresponding response models may thus differ---the FTA rate of the decision rule is consistently estimated to be
    approximately 12--14\%.
  }
\label{fig:sensitivity-judge}
  \vspace{-3mm}
\end{figure}

As a second robustness check, we adapt the method of \citet{rosenbaum_1983a}
for assessing the sensitivity of estimated causal effects to an unobserved
binary covariate.
We specifically tailor their approach to offline policy evaluation.
At a high level, we assume there is an unobserved covariate $u \in \{0,1\}$ that
affects both a judge's decision (RoR or bail) and also the outcome conditional
on that action.
For example, $u$ might indicate that a defendant is sympathetic, and sympathetic
defendants may be more likely to be RoR'd and also more likely to appear at
their court proceedings. Our key assumption is that a judge's action is
ignorable given the observed covariates $x$ and the unobserved covariate $u$:
\begin{equation}
\label{eq:ignorability}
  (r(\text{RoR}), r(\text{bail})) \independent a \ \big | \, x,u.
\end{equation}
There are four key parameters in this framework:
(1) the probability that $u = 1$;
(2) the effect of $u$ on the judge's decision;
(3) the effect of $u$ on the defendant's likelihood of FTA if RoR'd;
and (4) the effect of $u$ on the defendant's likelihood of FTA if bail is set.
Our goal is to quantify the extent to which our estimate of $V^{\pi}$
changes as a function of these parameters.

Without loss of generality, we can write
\begin{equation}
\label{eq:selection}
  \Pr(a = \text{RoR} | u, x) = \text{logit}^{-1}\left(\gamma_x + u \alpha_x \right)
\end{equation}
for appropriately chosen parameters $\gamma_x$ and $\alpha_x$ that depend on the
observed covariates $x$.
We note that randomness in judicial decisions may arise from a multitude of factors, including
idiosyncrasies in how judges are assigned to cases.
Here $\alpha_x$ is the change in log-odds of being RoR'd when $u=0$ versus when
$u=1$. For $t \in \{\text{RoR}, \text{bail}\}$, we can similarly write
\begin{equation}
\label{eq:effect}
  \Pr(r(t) | u, x) = \text{logit}^{-1}\left(\beta_x^t + u \delta_x^t \right)
\end{equation}
for parameters $\beta_x^t$ and $\delta_x^t$.
In this case, $\delta_x^{\text{RoR}}$ is the change in log-odds of failing to
appear if RoR'd when $u=0$ versus when $u=1$,
and $\delta_x^{\text{bail}}$ is the corresponding change if bail is set.

Now, for any posited values of $\Pr(u=1|x)$, $\alpha_x$, $\delta_x^{\text{RoR}}$
and $\delta_x^{\text{bail}}$, we use the observed data to estimate
$\gamma_x$, $\beta_x^{\text{RoR}}$ and $\beta_x^{\text{bail}}$.
We do this in three steps.
By~\eqref{eq:selection},
\begin{align*}
  & \Pr(a = \text{RoR} | x) =
    \Pr(u = 0|x) \cdot \text{logit}^{-1}(\gamma_x) + \\
  & \hspace{5mm}  \Pr(u=1|x) \cdot \text{logit}^{-1}(\gamma_x + \alpha_x).
\end{align*}
The left-hand side of the equation can be estimated with a regression model fit
to the data.
For fixed values of $\Pr(u=1|x)$ and $\alpha_x$, the right-hand side is an
increasing function of $\gamma_x$ that takes on values from 0 to 1 as $\gamma_x$
goes from $-\infty$ to $+\infty$.
There is thus a unique value $\hat{\gamma}_x$ such that the right-hand side
equals $\hat{\Pr}(a = \text{RoR} | x)$.
\citet{rosenbaum_1983a} derive a simple closed form
solution for $\hat{\gamma}_x$, facilitating fast computation on large datasets, which we omit for space.

Second, we use the fitted values of $\gamma_x$ to estimate the distribution of
$u$ given the observed covariates and judicial action. By Bayes' rule,
\begin{align*}
  & \Pr(u=1|a=t,x) = \frac{\Pr(a=t | u=1, x) \Pr(u=1|x)}{\Pr(a=t|x)} \\
  & = \frac{\Pr(a=t | u=1, x) \Pr(u=1|x)}
           {\Pr(a=t|u=1,x) \Pr(u=1|x) + \Pr(a=t| u=0,x) \Pr(u=0|x)}.
\end{align*}
With $\hat{\gamma}_x$, the $\Pr(a =t| u, x)$ terms on the right-hand side can be
estimated from~\eqref{eq:selection}, and we can thus approximate the left-hand
side.

Third, we have
\begin{align*}
  & \hspace{-3mm} \Pr(r(t) = 1 | a=t, x)  \\
  & = \Pr(u = 0|a=t, x)  \Pr(r(t) = 1 | a=t, x, u=0)   \\
  & \hspace{3mm} + \Pr(u = 1|a=t, x)  \Pr(r(t) = 1 | a=t, x, u=1)  \\
  & = \Pr(u = 0|a=t, x)  \Pr(r(t) = 1 | x, u=0)   \\
  & \hspace{3mm} + \Pr(u = 1|a=t, x)  \Pr(r(t) = 1 | x, u=1)  \\
  & = \Pr(u = 0|a=t, x)  \cdot \text{logit}^{-1}\left(\beta_x^t\right)   \\
  & \hspace{3mm} + \Pr(u = 1|a=t, x) \cdot \text{logit}^{-1}\left(\beta_x^t +
    \delta_x^t \right).
\end{align*}
The second equality above follows from the ignorability assumption stated
in~\eqref{eq:ignorability}, and the third equality follows
from~\eqref{eq:effect}.
The left-hand side can be approximated by the quantity $\hat{r}_x(t)$
that we obtain via response surface modeling.
Importantly, $\hat{r}_x(t)$ is a reasonable estimate of $\Pr(r(t) = 1 | a=t, x)$
even though it may not be a good estimate of $r_x(t)$.
This distinction is indeed the rationale of our sensitivity analysis.
Given our above estimate of $\Pr(u=1|a=t,x)$ and our assumed value of
$\delta_x^t$, the only unknown on the right-hand side is $\beta_x^t$. As before,
there is a unique value $\hat{\beta}_x^t$ that satisfies the constraint.

With $\hat{\beta}_x^t$ in hand, we can now approximate the potential outcome for
the action \emph{not} taken:
\begin{equation*}
  \Pr(r(\bar{t})= 1 | a=t, x)
\end{equation*}
where $\bar{t} = \text{RoR}$ if $t = \text{bail}$, and vice versa.
Specifically, we have
\begin{align}
  & \hat{\Pr}(r(\bar{t})= 1 | a=t, x) = \hat{\Pr}(u = 0|a=t, x)  \cdot
    \text{logit}^{-1}\left(\hat{\beta}_x^{\bar{t}}\right) + \nonumber \\
  & \hspace{5mm} \hat{\Pr}(u = 1|a=t, x) \cdot
    \text{logit}^{-1}\left(\hat{\beta}_x^{\bar{t}} +
    \delta_x^{\bar{t}} \right).
  \label{eq:counterfactual-estimate}
\end{align}
Finally, the Rosenbaum and Rubin estimator adapted to policy evaluation is
\begin{equation*}
  \hat{V}^{\pi}_{\text{RR}} = \frac{1}{|\Omega|} \sum_i \big [
    r_i \mathbf{I}(\pi(x_i) = a_i) + \hat{r}_i (\bar{a_i})
    \mathbf{I}(\pi(x_i) \neq a_i)
  \big ],
\end{equation*}
where $\hat{r}_i(\bar{a_i}) = \hat{\Pr}(r(\bar{a_i})= 1 | a_i, x_i)$ is computed
via~\eqref{eq:counterfactual-estimate}.

\begin{figure}
  \includegraphics[width=4cm]{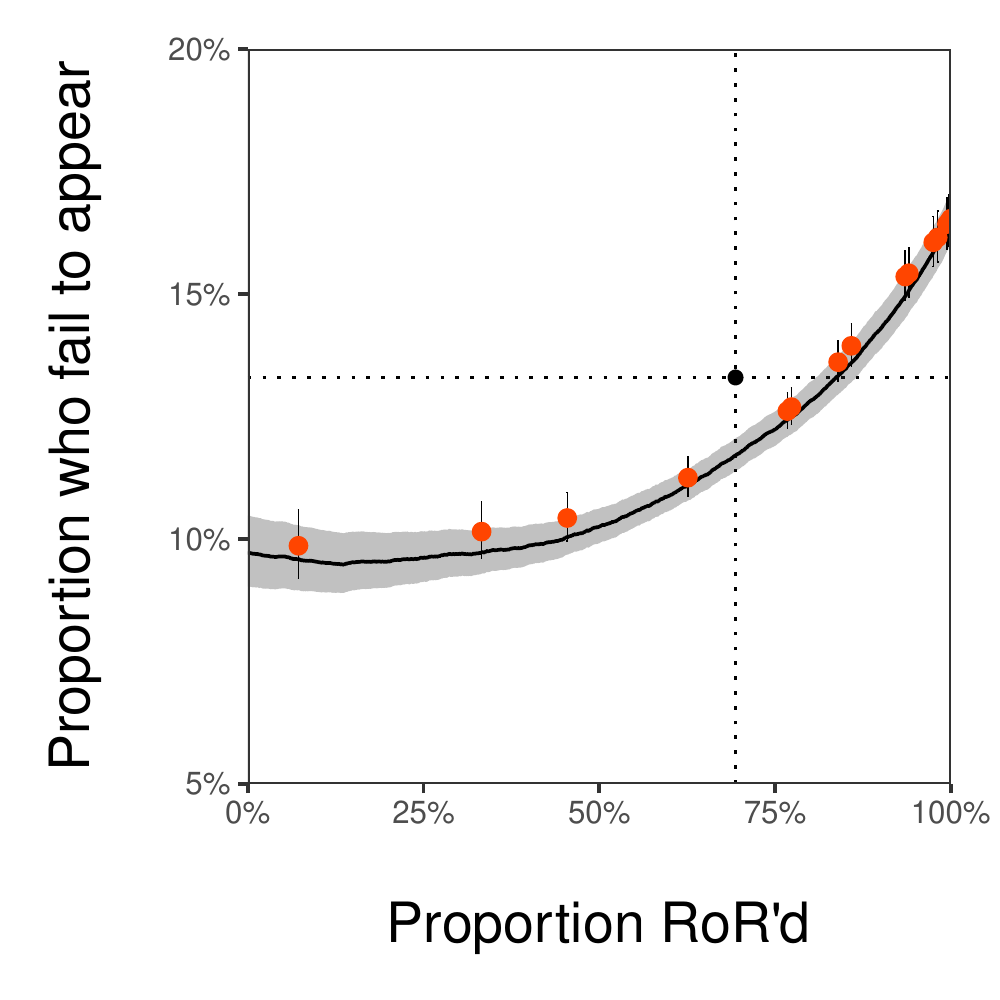}
  \includegraphics[width=4cm]{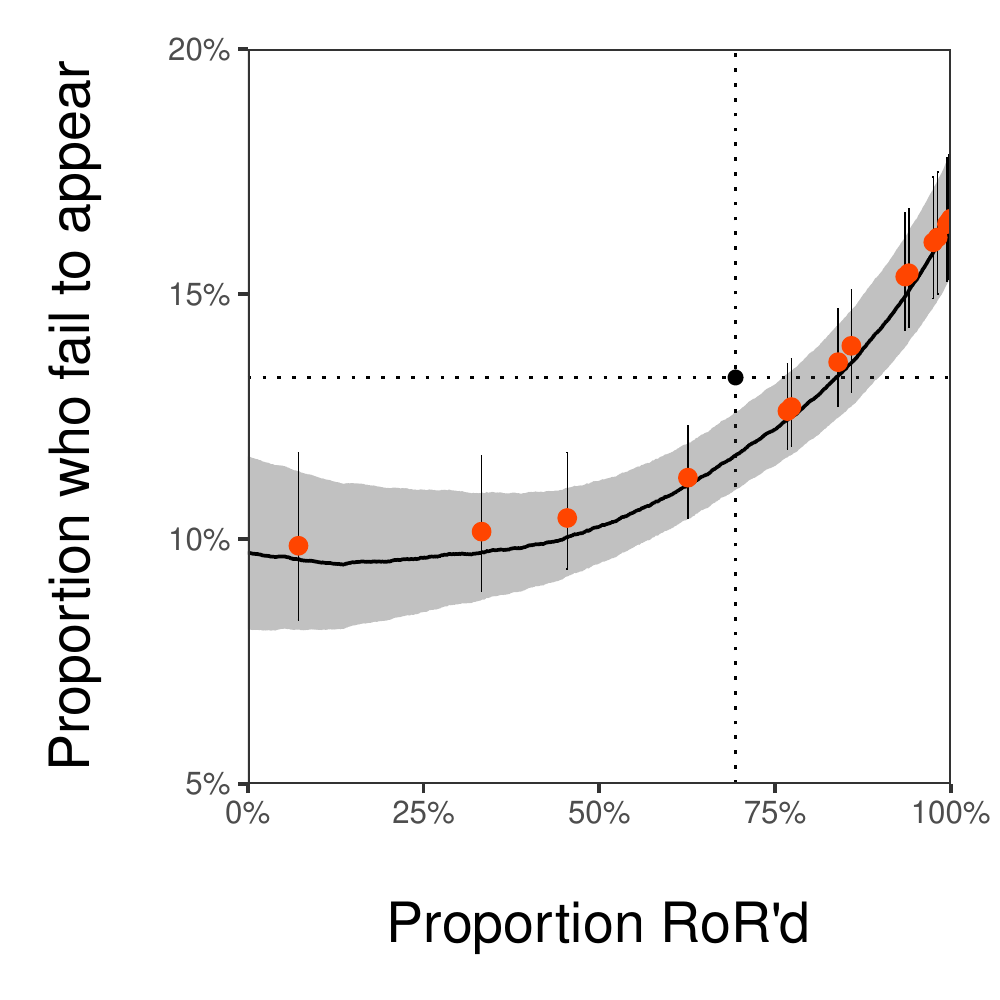}
   \caption{%
    The grey band (for the complex rules) and the error bars (for the simple
    rules) indicate minimum and maximum FTA estimates for a variety of parameter
    settings.  In the left-hand plot, we assume $\alpha = \log 2$ and consider
    all combinations of $p(u=1) \in \{0.1, 0.2, \dots, 0.9\}$,
    $\delta^{\text{RoR}} \in \{-\log{2}, 0, \log 2\}$, and $\delta^{\text{bail}}
    \in \{-\log{2}, 0, \log 2\}$, where all parameters are constant independent
    of $x$.
    In the right-hand plot, we consider a more extreme situation, with $\alpha
    = \log 3$, $\delta^{\text{RoR}} \in \{-\log 3, 0, \log 3\}$, and
    $\delta^{\text{bail}} \in \{-\log 3, 0, \log 3\}$.  The results are relatively
    stable in these parameter regimes.
  }
\label{fig:sensitivity}
  \vspace{-4mm}
\end{figure}

Figure~\ref{fig:sensitivity} shows the results of computing
$\hat{V}^{\pi}_{\text{RR}}$ on our data in two parameter regimes.
In the first (left-hand plot), we
assume $\alpha = \log 2$ and consider all combinations of
$p(u=1) \in \{0.1, 0.2, \dots, 0.9\}$, $\delta^{\text{RoR}} \in \{-\log{2}, 0, \log 2\}$, and
$\delta^{\text{bail}} \in \{-\log{2}, 0, \log 2\}$.
All parameters are constant independent of $x$.
We thus assume that holding the observed covariates fixed,
a defendant with $u=1$ has twice the odds of being RoR'd as one with $u=0$,
and that $u$ can double or half the odds a defendant fails to appear.
For each complex policy (\emph{i.e.}, one based on a random forest),
the grey band shows the minimum and maximum value of $\hat{V}^{\pi}_{\text{RR}}$
across all parameters in this set; the error bars on the red points show the
analogous quantity for the simple rules.  In the right-hand plot, we consider
a more extreme situation, with $\alpha = \log 3$, $\delta^{\text{RoR}} \in
\{-\log 3, 0, \log 3\}$, and $\delta^{\text{bail}} \in \{-\log 3, 0, \log 3\}$.
We find that our estimates are relatively stable in these parameter regimes.
In the first case ($\alpha = \log 2$) the estimated FTA rate for a given policy
typically varies by only half a percentage point.  Even in the more extreme
setting ($\alpha = \log 3$), policies are typically stable to about one
percentage point.  It thus seems our conclusions are robust to unobserved
heterogeneity across defendants.

\section{Select-regress-and-round: A simple method for creating simple rules}

We now introduce and evaluate a simple method---\emph{%
select-regress-and-round}---that formalizes
and generalizes the rule construction procedure we applied for pretrial release decisions.
In particular, we dispense with ad hoc feature selection and adopt a standard statistical
routine.

\subsection{Rule construction}
The rules we construct are designed to aid classification or ranking decisions
by assigning each item in consideration a score $z$, computed as a linear
combination of a subset $S$ of the item features:
\begin{equation*}
  z = \sum_{j \in S} w_j x_j, \label{eq:heuristic_score}
\end{equation*}
where the weights $w_j$ are integers. In the cases we consider, the features
themselves are typically 0\--1 indicator variables (indicating, for example,
whether a person is male, or whether an individual is 26--30 years old), and so
the rule reduces to a weighted checklist, in which one simply sums up the
(integer) weights of the applicable attributes. Often, one seeks to make binary
decisions (\emph{e.g.}, whether to detain or to release an individual), which amounts to
setting a threshold and then taking a particular course of action if and only if
the score is above that threshold.

This class of rules has two natural dimensions of complexity:
the number of features and the magnitude of the weights.
Given integers $k \geq 1$ and $M \geq 1$, we apply the following three-step
procedure to construct rules with at most $k$ features and integer weights
bounded by $M$ (\emph{i.e.}, $|S| \leq k$ and $-M \leq w_j \leq M$).

\begin{enumerate}[align=left,leftmargin=*]
  \item \textbf{Select.}
    From the full set of features, select $k$ features via forward stepwise
    regression.
      For fixed $k$, we note that standard selection metrics (\emph{e.g.}, AIC or BIC)
    are theoretically guaranteed to yield the same set of features.

  \item \textbf{Regress.}
    Using only these $k$ selected features, train an $L^1$-regularized (lasso)
    logistic regression model to the data, which yields (real-valued) fitted
    coefficients $\beta_1, \dots, \beta_k$.

  \item \textbf{Round.}
    Rescale the coefficients to be in the range $[-M, M]$, and then
    round the rescaled coefficients to the nearest integer. Specifically, set
    \begin{equation*}
      w_j = \text{Round}\left ( \frac{M \beta_j}{\max_i \vert \beta_i \vert} \right ).
    \end{equation*}
\end{enumerate}
We note that rules constructed in this way may have fewer than $k$ features, since the
  lasso regression in Step 2 may result in coefficients that are identically
  zero, and rescaling and rounding coefficients in Step 3 may zero-out
additional terms.\footnote{%
We select features in Step 1 with the \texttt{R} package \texttt{leaps}.
The models in Step 2 are fit with the \texttt{R} package \texttt{glmnet}. The \texttt{cv.glmnet} method is used to determine the best value of the regularization parameter $\lambda$ with 10-fold cross-validation and 1,000 values of $\lambda$.
}
This select-regress-and-round strategy for rule construction builds upon findings that  ``improper'' weighting schemes for linear models (e.g, unit weighting) lead to accurate predictions~\citep{guilford1942fundamental,dawes1979robust, gigerenzer1996reasoning,goel_2016c};
in particular, our strategy incorporates feature selection and more general integer weights to generate
a richer family of simple rules.
We next examine the accuracy of these rules.

\subsection{Rule evaluation}

We apply the select-regress-and-round procedure to 22 publicly available
datasets to examine the tradeoff between complexity and performance.
These datasets all come from the UCI ML repository,
and were selected according to four criteria:
(1) the dataset involves a binary classification (as opposed to a regression)
problem;\footnote{%
  For those datasets whose outcome variable takes more than two values, we set
  the majority class as the target variable, so that all the tasks we consider
involve binary classification.}
(2) the dataset is provided in a standard and complete form;
(3) the dataset involves more than 10 features;
and (4) the classification problem is one that a human could plausibly learn to
solve with the given features.
For example, we included a dataset in which the task was to determine whether
cells were malignant or benign based on various biological attributes of the
cells, but we excluded image recognition tasks in which the features were
represented as pixel values.  This fourth requirement limits the scope of our
analysis and conclusions to domains in which human decision makers typically act
without the aid of a computer.\footnote{The 22 UCI datasets we consider are:
adult, annealing, audiology-std, bank, bankruptcy, car, chess-krvk, chess-krvkp,
congress-voting, contrac, credit-approval, ctg, cylinder-bands, dermatology,
german\_credit, heart-cleveland, ilpd, mammo, mushroom, aus\_credit, wine, and
wine\_qual.
}

Unlike the judicial decisions discussed in Section~\ref{sec:bail},
outcomes in the domains we consider here are unaffected by a decision maker's actions.
For example, assessing the likelihood a cell is malignant---and then acting on that knowledge---does not change the fact that the cell was either malignant or not at the time of the measurement.
In contrast, a judge's decision to release or detain an individual necessarily alters the defendant's likelihood of appearing at trial.
Further, in the UCI domains, we observe outcomes for every example, not only a subset in which a decision maker chose to act.
Decision rules are constructed similarly in both the UCI and bail datasets.
Evaluating the resulting rules, however, is significantly easier for the UCI datasets:
since outcomes are independent of actions and are observed for all examples,
one need not consider subtle issues of causal inference.

\begin{figure}[t]
  \includegraphics[width=6cm]{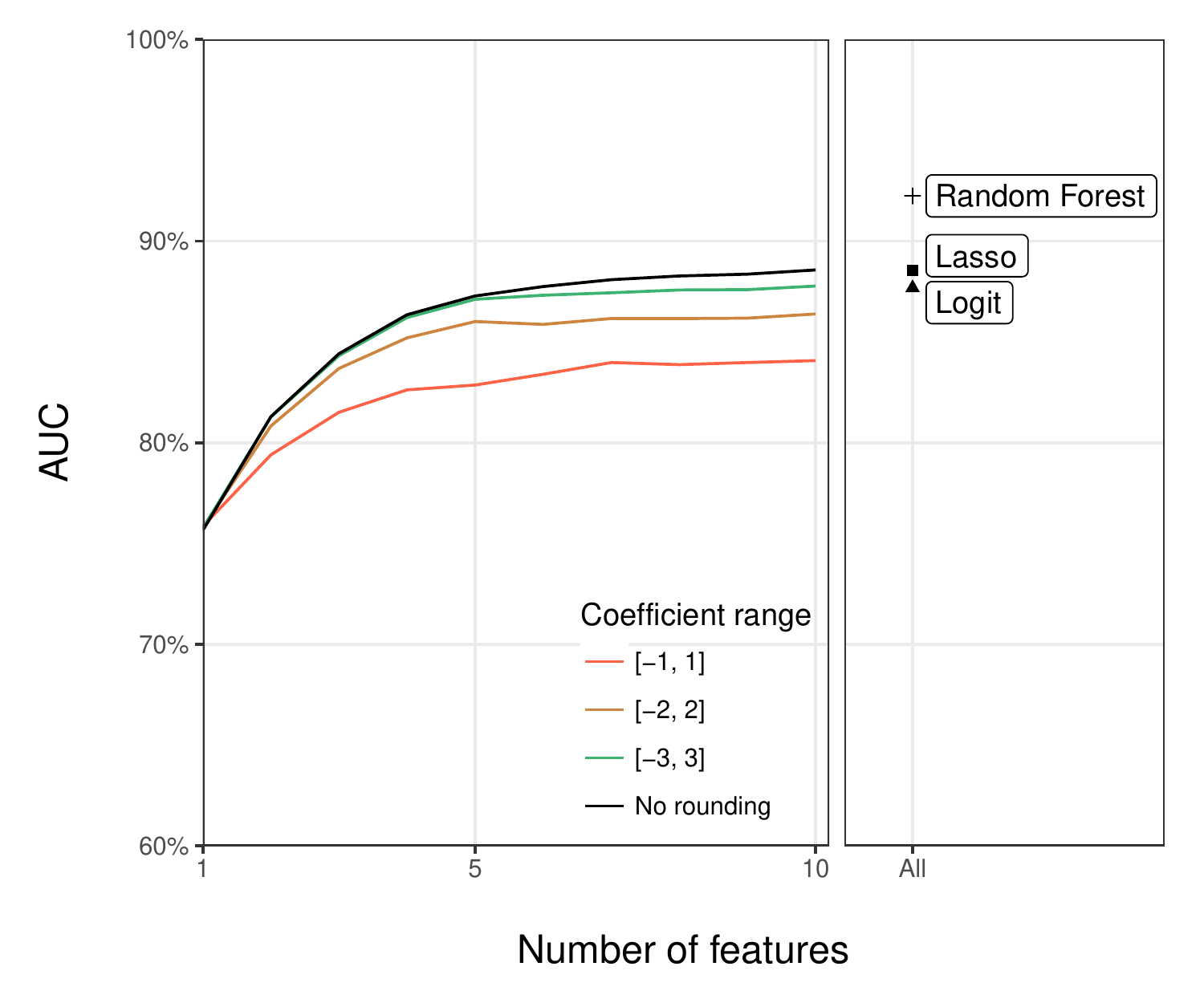}
  \caption{%
    Mean test AUC of decision rules over 22 datasets.
    The simple rules use up to 10 features,
    with integer coefficients in the specified ranges.
    The black line shows performance of lasso with feature selection but
    without rounding the coefficients.
    ``All'' features --- used by random forest, lasso, and logistic
    regression --- varies by domain, with an average of 38.
  }
\label{fig:auc_no_noise}
  \vspace{-5mm}
\end{figure}

On each of the 22 datasets we analyze here, we construct simple rules for
a range of the number of features $k \in \{1, \dots, 10\}$ and the magnitude of
the weights $M \in \{1, 2, 3\}$. We benchmark the performance of these rules
against three standard statistical models: logistic regression,
$L^1$-regularized logistic regression, and random forest.  These models were fit
in \texttt{R} with the \texttt{glm},  \texttt{glmnet}, and \texttt{randomForest}
packages, respectively.  For the $L^1$-regularized logistic regression models,
the \texttt{cv.glmnet} method was used to determine the best value of the
regularization parameter $\lambda$ with 10-fold cross-validation and 1,000
values of $\lambda$.
We used 1,000 trees for the random forest models.
This head-to-head comparison is a difficult test for the simple rules in part
because they can only base their predictions on 1 to 10 features.
The complex models, in contrast, can train and predict with all features,
which number between 11 and 93 with a mean of 38.

Figure~\ref{fig:auc_no_noise} shows model performance---measured in terms of
mean AUC across the 22 datasets---as a function of model size and coefficient
range.
The AUC for each model on each dataset is computed via 10-fold cross-validation.
We find that simple rules with only five features and integer coefficients
between -3 and 3 perform on par with logistic regression and $L^1$-regularized
logistic regression trained on the full set of features.
For 1 to 10 features, the [-3, 3] model (green line) differs from the
unrounded lasso model (black line) by less than 1 percentage point.
The performance of the random forest model is somewhat better:
trained on all features, random forest achieves mean AUC of 92\%;
the mean AUC is 87\% for simple rules with at most five features and integer
coefficients between -3 and 3.
Complex prediction methods certainly have their advantages, but the gap in
performance between simple rules and fully optimized prediction methods
is not as large as one might have thought.

\subsection{Benchmarking to integer programming}
The simple rules we construct take the form of a linear scoring rule with
integer weights.
To produce such rules, mixed-integer programming is a natural alternative to
our select-regress-and-round strategy, and supersparse linear integer models
(SLIM)~\cite{ustun_2016} is the leading instantiation of that approach.
Given constraints on the number of features and the magnitude of the integer
weights, SLIM produces rules that optimize for binary classification accuracy
(\emph{i.e.}, 0\--1 loss).

\begin{figure}[t]
  \includegraphics[width=4cm]{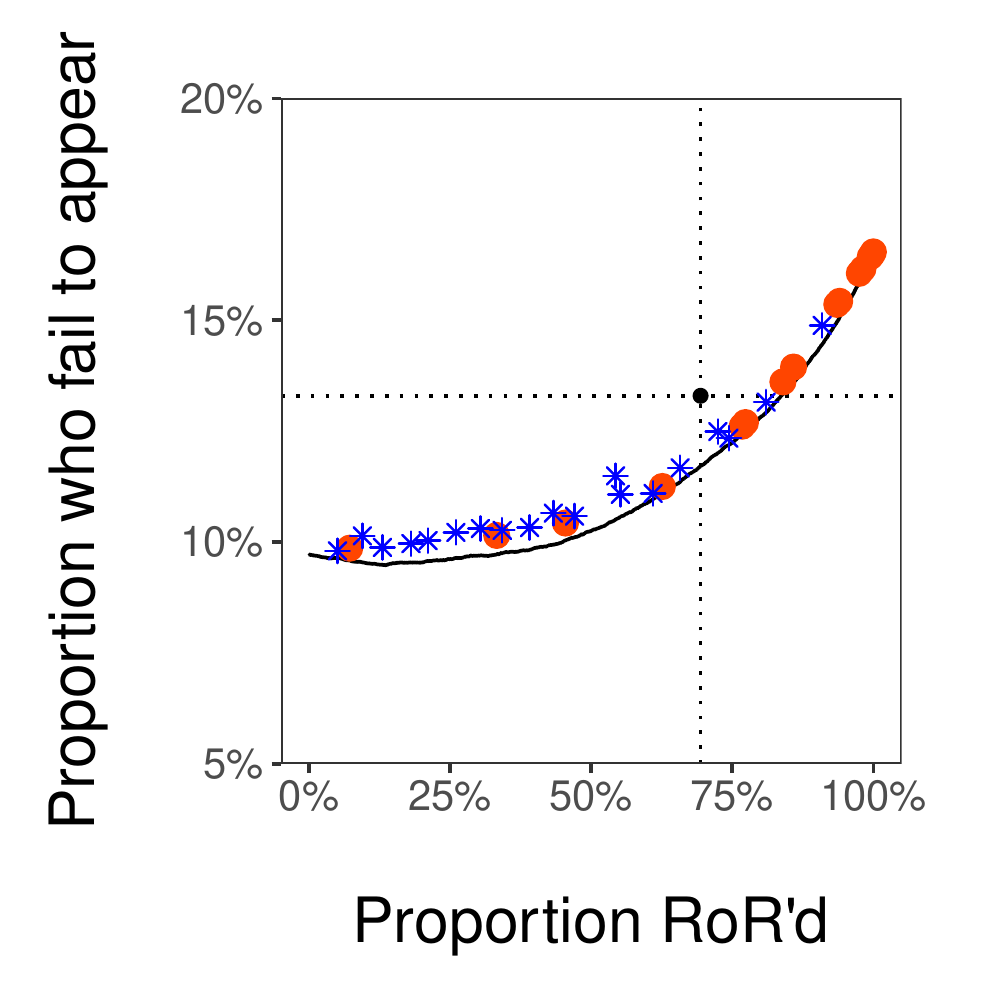}
  \includegraphics[width=4cm]{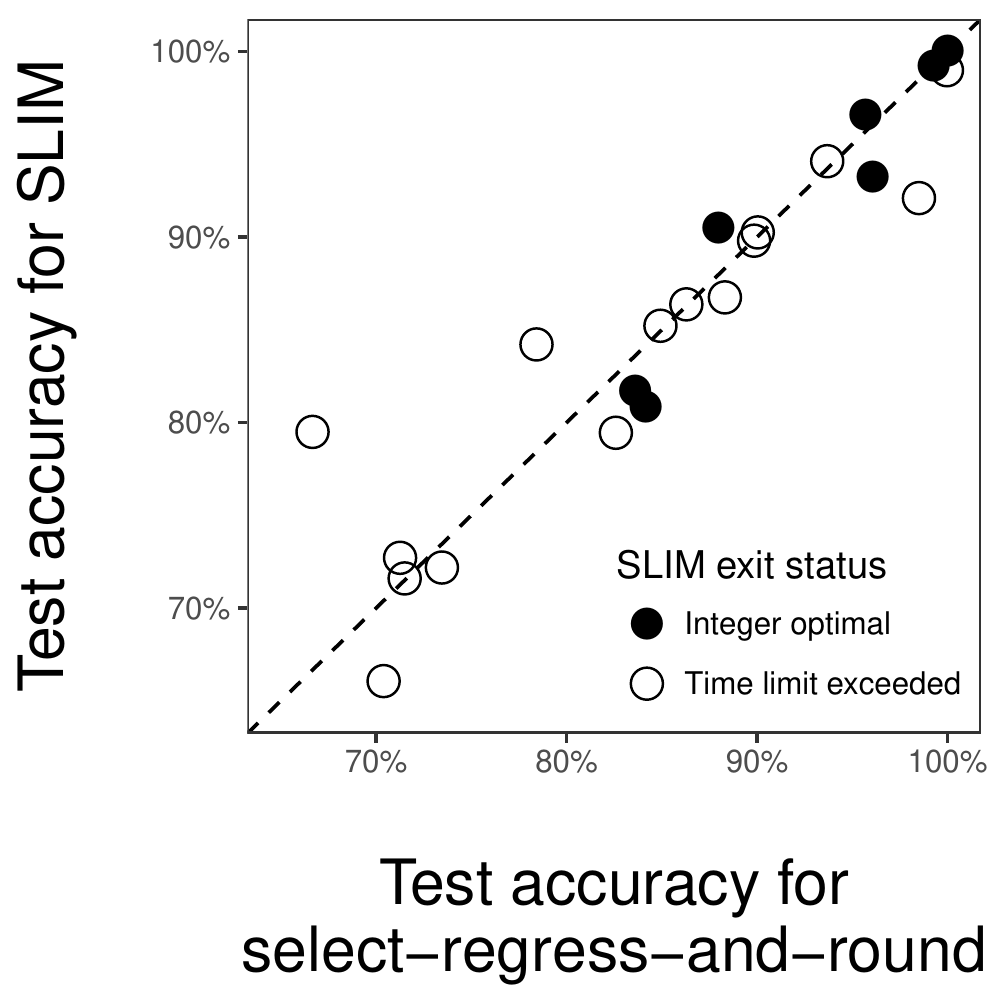}
   \caption{%
     Left panel: comparison of rules for pretrial release decisions
     produced by select-regress-and-round (red), SLIM (blue), and
     random forest (black line).
     Right panel: binary classification accuracy for
     select-regress-and-round and SLIM on 22 UCI datasets.
  }
\label{fig:compare_to_slim}
  \vspace{-6mm}
\end{figure}

We compare SLIM to select-regress-and-round on the judicial decision-making
problem and on the 22 UCI datasets.
Figure~\ref{fig:compare_to_slim} (left panel) shows estimated FTA and release
rates for the random forest model (black line), our simple rules derived in
Section~\ref{sec:bail} (red points), and the simple rules produced by SLIM (blue
points).
As with our own simple rules, we constrain SLIM to produce rules based on age
and number of past FTAs, with integer weights ranging from -10 to 10.
As before, decision rules are constructed from
the random forest and select-regress-and-round risk scores by
varying the decision threshold;
in contrast, multiple rules for SLIM are computed by varying a parameter that specifies the maximum acceptable false positive rate~\cite{ustun_2016}.
Both methods for producing simple rules perform nearly the same
as the random forest model trained on the full set of 64 features.

We next consider the 22 UCI datasets.
SLIM is known to work best when the features are
discrete~\cite{zeng2016interpretable}.
We thus pre-process the datasets by discretizing all continuous features into
three bins containing an approximately equal number of examples, representing
low, medium, and high values of the feature.
Integer programming is an NP-hard problem, and so following
\citet{ustun_2016} we set a time limit for SLIM;\@ they set a 10-minute limit,
but we allow up to 6 hours of computation per model.
For 7 of the 22 datasets, SLIM found an integer-optimal solution within
the time limit, and returned approximate solutions in the remaining 15 cases.
Figure~\ref{fig:compare_to_slim} (right panel)
compares binary classification accuracy of SLIM and select-regress-and-round on
the 22 UCI datasets, where each point corresponds to a dataset.
Both methods are constrained to produce rules with at most five features and
integer coefficients between -3 and 3.
We show 0\--1 accuracy since SLIM optimizes for this metric, but similar results
hold for AUC;\@ accuracy is computed out-of-sample via 10-fold
cross-validation.
Both methods for producing simple rules yield comparable results.
Averaged across all 22 datasets, SLIM and select-regress-and-round both achieve mean
accuracy of 86\%.
Even in the 7 cases where SLIM found integer-optimal solutions,
performance is nearly identical to our simple select-regress-and-round strategy.

In terms of classification accuracy, select-regress-and-round generates rules on
par with those obtained by solving mixed-integer programs.
We note, however, two advantages of our approach.
First, whereas select-regress-and-round yields results almost instantaneously,
integer programs can be computationally expensive to solve.
Second, our approach is both conceptually and technically simple,
requiring little statistical or computational expertise,
and accordingly easing adoption for practitioners.

\section{The robustness of classification}

Why is it that simple rules often perform as well as the most sophisticated
statistical methods? In part it is because binary classification is robust to
error in the underlying predictive model, an observation that we formalize in
Theorem~\ref{theorem:auc} below.

To establish this result, we start by considering the prediction scores
generated via a standard statistical method---such as logistic regression
trained on the full set of available features---which we call the ``true''
scores. As in linear discriminant analysis, we assume that the true scores for positive and negative instances are
normally distributed with equal variance:
$\text{N}(\mu_p, \sigma^2)$ and $\text{N}(\mu_n, \sigma^2)$, respectively.
The homoscedasticity assumption guarantees the Bayes optimal classifier is a threshold rule on the scores.
For scores estimated via logistic regression, the normality assumption is reasonable if we consider the
scores on the logit scale rather than on the probability scale.
Figure~\ref{fig:empirical_gamma_app} (left panel) shows such scores for one of the UCI
datasets.
We further assume that the process of generating simple rules---both limiting
the number of features and also restricting the possible values of the
weights---can be viewed as adding normal, mean-zero noise $\text{N}(0,
\sigma_{\epsilon}^2)$ to the true scores; Figure~\ref{fig:empirical_gamma_app} (center panel) plots the distribution of this noise for one of the datasets.\footnote{%
\label{fn:noise}
 We estimate the noise distribution by
  taking the difference between the simple and true scores.
  Before taking the difference, we convert the simple scores to the scale of
  true scores by dividing the simple scores by $M$, the scaling factor used when generating the rule.
  }
Thus, with simple rules, instead of making classification decisions based on the
true scores, we assume decisions are made in terms of a noisy approximation.
Under this analytic framework, Theorem~\ref{theorem:auc} shows that the drop in
classification performance (as measured by AUC) can be expressed in terms of the
``true AUC'' (\emph{i.e.}, the AUC under the true scores) and $\gamma
= \sigma_{\epsilon}^2/\sigma^2$, the ratio of the noise to the
within-class variance of the true scores.
In particular, we find that when the magnitude of the noise is on par with (or
smaller than) the score variance (\emph{i.e.}, $\gamma \lesssim 1$),
then the AUC of the noisy approximation is comparable to the true AUC.\@

\begin{theorem}
For a binary classification task, let $Y$ be a continuous random variable that
denotes the prediction score of a random instance, and let $Y_p$ and $Y_n$
denote the conditional distributions of $Y$ for positive and negative instances,
respectively.
Suppose $Y_p \sim \text{N}(\mu_p, \sigma^2)$ and $Y_n \sim \text{N}(\mu_n,
\sigma^2)$.
Then, for $\epsilon \sim \text{N}(0, \sigma_{\epsilon}^2)$ and $\hat{Y} = Y + \epsilon$,
\begin{align}
\label{eq:auc}
    \mathrm{AUC}_{\hat{Y}} & =
    \Phi \left( \frac{\Phi^{-1}(\mathrm{AUC}_Y)}{\sqrt{1+\gamma}}\right),
\end{align}
where $\gamma = \sigma_{\epsilon}^2/\sigma^2$, and $\Phi$ is the CDF for
the standard normal.
\label{theorem:auc}
\end{theorem}

\begin{proof}
    In general, AUC is equal to the probability that a randomly selected positive instance has a
    higher prediction score than a randomly selected negative instance, and so
    $\text{AUC}_Y = \Pr(Y_p - Y_n > 0)$.
    Since $Y_{p} - Y_{n}$ is normally distributed with mean $\mu_{p} - \mu_{n}$
    and variance $2\sigma^2$,
    \[
        \frac{Y_{p} - Y_{n} - (\mu_{p} - \mu_{n})}
        {\sqrt{2}\sigma} \sim \text{N}(0,1).
    \]
    Hence,
    \begin{align*}
        \text{AUC}_{Y} &= \Pr\left(
          \frac{Y_{p} - Y_{n} - (\mu_{p} - \mu_{n})}
          {\sqrt{2}\sigma} >
          -\frac{\mu_p - \mu_n}{\sqrt{2}\sigma}\right) \\
        & = \Phi\left(\frac{\mu_p - \mu_n}
        {\sqrt{2}\sigma}\right),
  \end{align*}
   where the last equality follows from symmetry of the normal distribution.

   Now define $\hat{Y}_{p} = Y_{p} + \epsilon$, so $\hat{Y}_{p} \sim \text{N}
   (\mu_{p}, \sigma^2 + \sigma_{\epsilon}^2)$, with $\hat{Y}_{n}$ defined
   similarly. A short computation shows that
    \begin{align*}
        \text{AUC}_{\hat{Y}}
        & = \Pr(\hat{Y}_p  > \hat{Y}_n)
        = \Phi\left(\frac{\mu_p - \mu_n}{\sqrt{2\sigma^2
            + 2\sigma_{\epsilon}^2}}\right)
        = \Phi \left(
          \frac{\Phi^{-1}(\mathrm{AUC}_Y)}{\sqrt{1+\gamma}}\right).
    \end{align*}
\end{proof}

Theorem~\ref{theorem:auc} establishes a direct theoretical link between
performance and noise in model specification.
To give a better sense of how the analytic expression for
$\mathrm{AUC}_{\hat{Y}}$ varies with $\mathrm{AUC}_Y$ and $\gamma$,
Figure~\ref{fig:empirical_gamma_app} (right panel) shows this expression for various parameter
values.
For example, the figure shows that for $\text{AUC}_Y = 90\%$ and
$\gamma = 0.5$, we have $\text{AUC}_{\hat{Y}} = 85\%$.
That is, if the amount of noise is equal to half the within-class variance of the
true scores, then the drop in performance is relatively small.

While connecting
model performance to model noise, Theorem~\ref{theorem:auc} leaves unanswered how much
noise simple rules add to the underlying scores.
This question seems
difficult to answer theoretically.
We can, however, empirically estimate
how much noise simple rules add in the datasets we analyze.\footnote{%
To estimate $\gamma = \sigma_{\epsilon}^2/\sigma^2$ for
  a specific simple rule on a given dataset, we first compute the average within-class variance
  of the true scores, where these scores are generated via an $L^1$-regularized logistic
  regression model.
  We estimate $\sigma_{\epsilon}^2$ by taking the variance of
  the noise, as described in Footnote~\ref{fn:noise}.
}
Across the 22 UCI datasets we consider, we find that rules with five features and a coefficient range of -3 to 3 have an average value of $\gamma = 0.22$.
This low empirically observed noise is in line with our finding that such simple rules perform well on these datasets.

\begin{figure}
    \includegraphics[width=2.7cm]{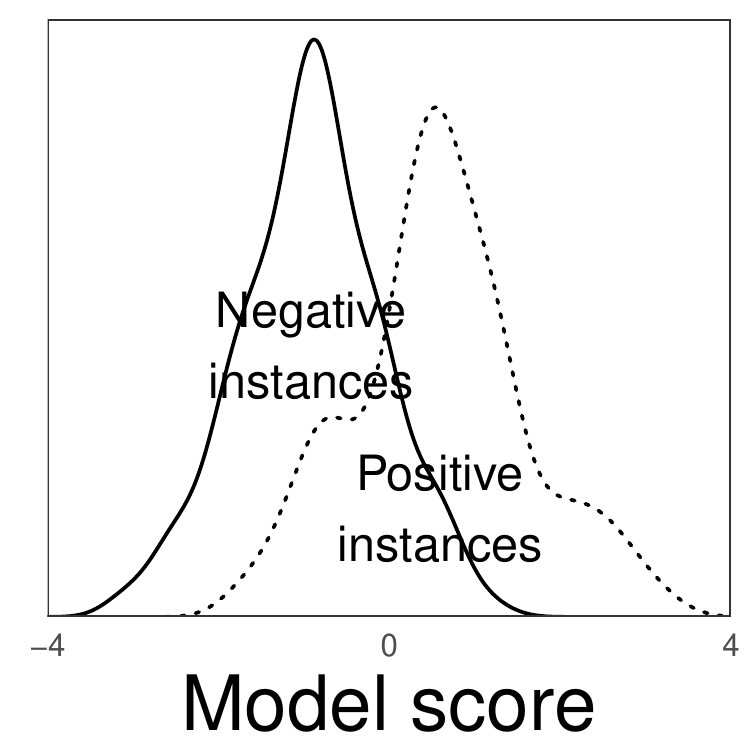}
    \includegraphics[width=2.7cm]{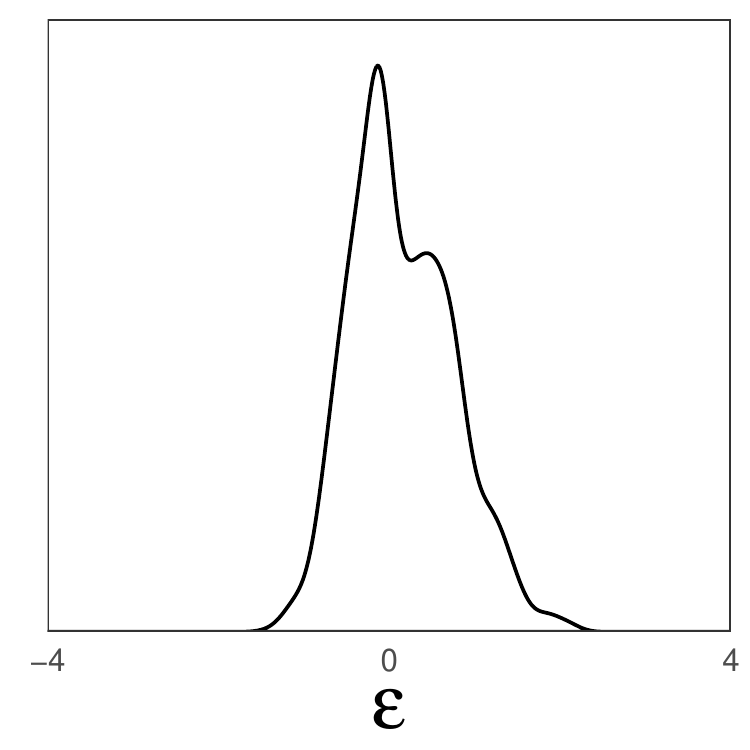}
    \includegraphics[width=2.7cm]{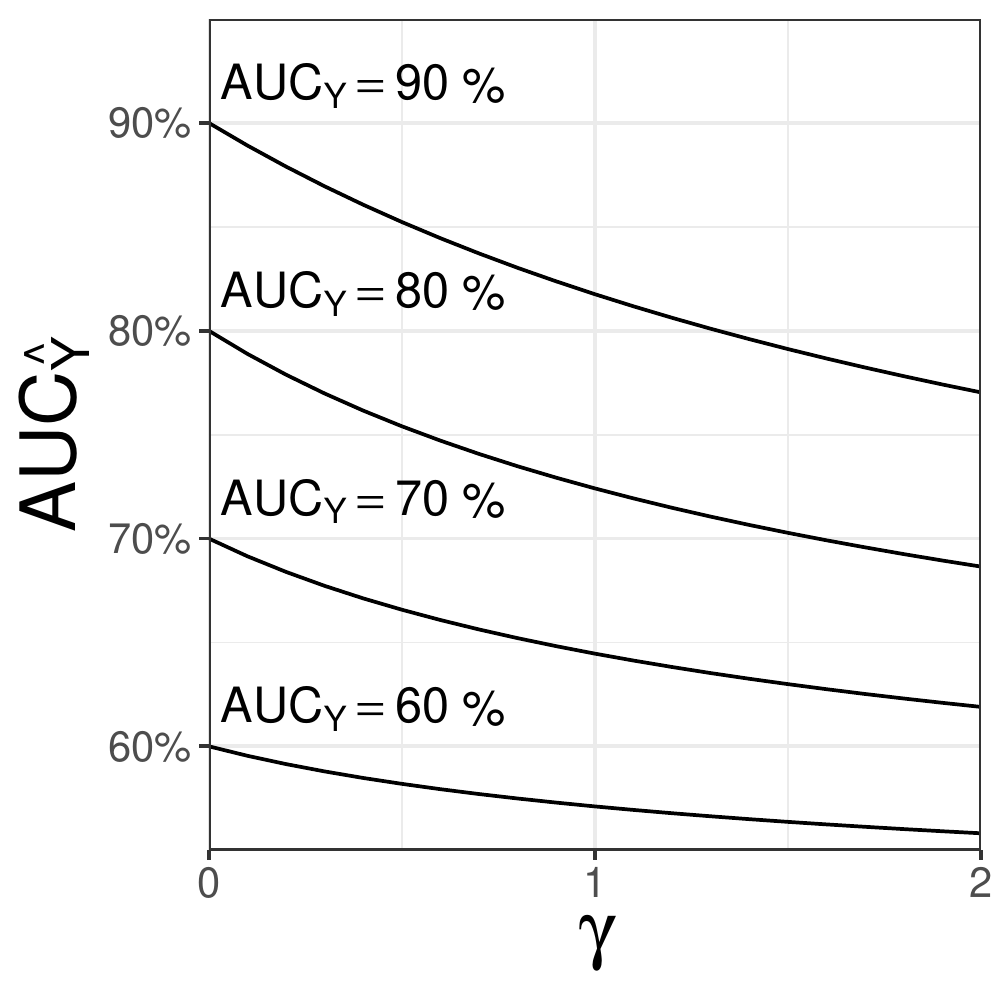}
 \caption{%
  Left panel: empirical distribution of prediction scores, on the logit scale,
    for positive and negative instances of a UCI dataset
    (\texttt{heart-cleveland}), generated via an $L^1$-regularized logistic
    regression model.
  Center panel: empirical distribution of $\epsilon$ for
    select-regress-and-round applied to the same dataset.
  Right panel: the theoretical change in AUC, under the setup of
    Theorem~\ref{theorem:auc}.
  }
\label{fig:empirical_gamma_app}
  \vspace{-4mm}
\end{figure}

\section{Conclusion}

In this paper we introduced select-regress-and-round, a simple method for
constructing decision rules that are fast, frugal, and clear. In an analysis of
pretrial release decisions, simple rules outperformed human judges
and matched the performance of a sophisticated statistical model.
Generalizing this result, in 22 domains of varying size and complexity, the
simple mental checklists produced by the select-regress-and-round method rivaled
the performance of regularized regression models while using only a fraction of
the information.

These results complement a growing body of work in statistics and computer
science in which sophisticated algorithms are used to create interpretable
scoring systems and rule sets~\citep{ustun_2016,letham_2015,lakkaraju_2016,lakkaraju_2017}.
Many prior rule
construction methods offer great flexibility and performance~\cite{ustun_2017}, but in turn
require considerable computational expertise to carry out.
In contrast, the simple rules in this article can be created by practitioners
with only basic statistical knowledge and generic software.
For practitioners
to favor statistics over intuition,
we believe decision rules must not only be simple to apply
but also simple to create.

\begin{acks}
We thank Avi Feller, Andrew Gelman, Gerd Gigerenzer, Art Owen, and Berk Ustun
for helpful conversations.
\end{acks}

\bibliographystyle{ACM-Reference-Format}
\bibliography{mainbib,Mendeley_Simple_rules_DANY}


\begin{thebibliography}{00}


\ifx \showCODEN    \undefined \def \showCODEN     #1{\unskip}     \fi
\ifx \showDOI      \undefined \def \showDOI       #1{{\tt DOI:}\penalty0{#1}\ }
  \fi
\ifx \showISBNx    \undefined \def \showISBNx     #1{\unskip}     \fi
\ifx \showISBNxiii \undefined \def \showISBNxiii  #1{\unskip}     \fi
\ifx \showISSN     \undefined \def \showISSN      #1{\unskip}     \fi
\ifx \showLCCN     \undefined \def \showLCCN      #1{\unskip}     \fi
\ifx \shownote     \undefined \def \shownote      #1{#1}          \fi
\ifx \showarticletitle \undefined \def \showarticletitle #1{#1}   \fi
\ifx \showURL      \undefined \def \showURL       #1{#1}          \fi
\providecommand\bibfield[2]{#2}
\providecommand\bibinfo[2]{#2}
\providecommand\natexlab[1]{#1}
\providecommand\showeprint[2][]{arXiv:#2}

\bibitem[\protect\citeauthoryear{Cassel, S{\"a}rndal, and Wretman}{Cassel
  et~al\mbox{.}}{1976}]%
        {cassel_1976}
\bibfield{author}{\bibinfo{person}{Claes~M Cassel}, \bibinfo{person}{Carl~E
  S{\"a}rndal}, {and} \bibinfo{person}{Jan~H Wretman}.}
  \bibinfo{year}{1976}\natexlab{}.
\newblock \showarticletitle{Some results on generalized difference estimation
  and generalized regression estimation for finite populations}.
\newblock \bibinfo{journal}{{\em Biometrika\/}} \bibinfo{volume}{63},
  \bibinfo{number}{3} (\bibinfo{year}{1976}), \bibinfo{pages}{615--620}.
\newblock


\bibitem[\protect\citeauthoryear{Corbett-Davies, Pierson, Feller, Goel, and
  Huq}{Corbett-Davies et~al\mbox{.}}{2017}]%
        {corbett-davies_2017}
\bibfield{author}{\bibinfo{person}{Sam Corbett-Davies}, \bibinfo{person}{Emma
  Pierson}, \bibinfo{person}{Avi Feller}, \bibinfo{person}{Sharad Goel}, {and}
  \bibinfo{person}{Aziz Huq}.} \bibinfo{year}{2017}\natexlab{}.
\newblock \showarticletitle{Algorithmic decision making and the cost of
  fairness}.
\newblock \bibinfo{journal}{{\em arXiv preprint arXiv:1701.08230\/}}
  (\bibinfo{year}{2017}).
\newblock


\bibitem[\protect\citeauthoryear{Danziger, Levav, and Avnaim-Pesso}{Danziger
  et~al\mbox{.}}{2011}]%
        {danziger_2011}
\bibfield{author}{\bibinfo{person}{Shai Danziger}, \bibinfo{person}{Jonathan
  Levav}, {and} \bibinfo{person}{Liora Avnaim-Pesso}.}
  \bibinfo{year}{2011}\natexlab{}.
\newblock \showarticletitle{Extraneous factors in judicial decisions}.
\newblock \bibinfo{journal}{{\em Proceedings of the National Academy of
  Sciences\/}} \bibinfo{volume}{108}, \bibinfo{number}{17}
  (\bibinfo{year}{2011}), \bibinfo{pages}{6889--6892}.
\newblock


\bibitem[\protect\citeauthoryear{Dawes}{Dawes}{1979}]%
        {dawes1979robust}
\bibfield{author}{\bibinfo{person}{Robyn~M Dawes}.}
  \bibinfo{year}{1979}\natexlab{}.
\newblock \showarticletitle{The robust beauty of improper linear models in
  decision making.}
\newblock \bibinfo{journal}{{\em American Psychologist\/}}
  \bibinfo{volume}{34}, \bibinfo{number}{7} (\bibinfo{year}{1979}),
  \bibinfo{pages}{571}.
\newblock


\bibitem[\protect\citeauthoryear{Dawes, Faust, and Meehl}{Dawes
  et~al\mbox{.}}{1989}]%
        {dawes1989clinical}
\bibfield{author}{\bibinfo{person}{Robyn~M Dawes}, \bibinfo{person}{David
  Faust}, {and} \bibinfo{person}{Paul~E Meehl}.}
  \bibinfo{year}{1989}\natexlab{}.
\newblock \showarticletitle{Clinical versus actuarial judgment}.
\newblock \bibinfo{journal}{{\em Science\/}} \bibinfo{volume}{243},
  \bibinfo{number}{4899} (\bibinfo{year}{1989}), \bibinfo{pages}{1668--1674}.
\newblock


\bibitem[\protect\citeauthoryear{Dhami}{Dhami}{2003}]%
        {Dhami2003}
\bibfield{author}{\bibinfo{person}{Mandeep~K Dhami}.}
  \bibinfo{year}{2003}\natexlab{}.
\newblock \showarticletitle{Psychological models of professional decision
  making}.
\newblock \bibinfo{journal}{{\em Psychological Science\/}}
  \bibinfo{volume}{14}, \bibinfo{number}{2} (\bibinfo{year}{2003}),
  \bibinfo{pages}{175--180}.
\newblock


\bibitem[\protect\citeauthoryear{Dud{\'{i}}k, Langford, and Li}{Dud{\'{i}}k
  et~al\mbox{.}}{2011}]%
        {dudik_2011}
\bibfield{author}{\bibinfo{person}{Miroslav Dud{\'{i}}k}, \bibinfo{person}{John
  Langford}, {and} \bibinfo{person}{Lihong Li}.}
  \bibinfo{year}{2011}\natexlab{}.
\newblock \showarticletitle{Doubly Robust Policy Evaluation and Learning}.
\newblock \bibinfo{journal}{{\em {ICML}\/}} (\bibinfo{year}{2011}).
\newblock
\showISBNx{978-1-4503-0619-5}
\showISSN{0883-4237}
\showDOI{%
\url{http://dx.doi.org/10.1214/14-STS500}}


\bibitem[\protect\citeauthoryear{Fern{\'a}ndez-Delgado, Cernadas, Barro, and
  Amorim}{Fern{\'a}ndez-Delgado et~al\mbox{.}}{2014}]%
        {fernandez2014}
\bibfield{author}{\bibinfo{person}{Manuel Fern{\'a}ndez-Delgado},
  \bibinfo{person}{Eva Cernadas}, \bibinfo{person}{Sen{\'e}n Barro}, {and}
  \bibinfo{person}{Dinani Amorim}.} \bibinfo{year}{2014}\natexlab{}.
\newblock \showarticletitle{Do we need hundreds of classifiers to solve real
  world classification problems?}
\newblock \bibinfo{journal}{{\em J. Mach. Learn. Res\/}} \bibinfo{volume}{15},
  \bibinfo{number}{1} (\bibinfo{year}{2014}), \bibinfo{pages}{3133--3181}.
\newblock


\bibitem[\protect\citeauthoryear{Gigerenzer and Goldstein}{Gigerenzer and
  Goldstein}{1996}]%
        {gigerenzer1996reasoning}
\bibfield{author}{\bibinfo{person}{Gerd Gigerenzer} {and}
  \bibinfo{person}{Daniel~G Goldstein}.} \bibinfo{year}{1996}\natexlab{}.
\newblock \showarticletitle{Reasoning the fast and frugal way: models of
  bounded rationality.}
\newblock \bibinfo{journal}{{\em Psychological review\/}}
  \bibinfo{volume}{103}, \bibinfo{number}{4} (\bibinfo{year}{1996}),
  \bibinfo{pages}{650}.
\newblock


\bibitem[\protect\citeauthoryear{Gigerenzer, Hertwig, and Pachur}{Gigerenzer
  et~al\mbox{.}}{2011}]%
        {Gigerenzer2011}
\bibfield{author}{\bibinfo{person}{Gerd Gigerenzer}, \bibinfo{person}{Ralph
  Hertwig}, {and} \bibinfo{person}{Thorsten Pachur}.}
  \bibinfo{year}{2011}\natexlab{}.
\newblock \bibinfo{booktitle}{{\em Heuristics: The foundations of adaptive
  behavior}}.
\newblock \bibinfo{publisher}{Oxford University Press, Inc.}
\newblock


\bibitem[\protect\citeauthoryear{Gleicher}{Gleicher}{2016}]%
        {gleicher2016}
\bibfield{author}{\bibinfo{person}{Michael Gleicher}.}
  \bibinfo{year}{2016}\natexlab{}.
\newblock \showarticletitle{A Framework for Considering Comprehensibility in
  Modeling}.
\newblock \bibinfo{journal}{{\em Big Data\/}} \bibinfo{volume}{4},
  \bibinfo{number}{2} (\bibinfo{year}{2016}), \bibinfo{pages}{75--88}.
\newblock


\bibitem[\protect\citeauthoryear{Goel, Rao, and Shroff}{Goel
  et~al\mbox{.}}{2016}]%
        {goel_2016c}
\bibfield{author}{\bibinfo{person}{Sharad Goel}, \bibinfo{person}{Justin~M
  Rao}, {and} \bibinfo{person}{Ravi Shroff}.} \bibinfo{year}{2016}\natexlab{}.
\newblock \showarticletitle{Precinct or Prejudice? {U}nderstanding Racial
  Disparities in {New York City's} Stop-and-Frisk Policy}.
\newblock \bibinfo{journal}{{\em Annals of Applied Statistics\/}}
  (\bibinfo{year}{2016}).
\newblock


\bibitem[\protect\citeauthoryear{Goodman and Flaxman}{Goodman and
  Flaxman}{2016}]%
        {goodman2016eu}
\bibfield{author}{\bibinfo{person}{Bryce Goodman} {and} \bibinfo{person}{Seth
  Flaxman}.} \bibinfo{year}{2016}\natexlab{}.
\newblock \showarticletitle{{EU} regulations on algorithmic decision-making and
  a right to explanation}.
\newblock \bibinfo{journal}{{\em arXiv preprint arXiv:1606.08813\/}}
  (\bibinfo{year}{2016}).
\newblock


\bibitem[\protect\citeauthoryear{Guilford}{Guilford}{1942}]%
        {guilford1942fundamental}
\bibfield{author}{\bibinfo{person}{Joy~Paul Guilford}.}
  \bibinfo{year}{1942}\natexlab{}.
\newblock \bibinfo{booktitle}{{\em Fundamental statistics in psychology and
  education.}}
\newblock \bibinfo{publisher}{McGraw-Hill}.
\newblock


\bibitem[\protect\citeauthoryear{Hill}{Hill}{2012}]%
        {hill2012}
\bibfield{author}{\bibinfo{person}{Jennifer~L Hill}.}
  \bibinfo{year}{2012}\natexlab{}.
\newblock \showarticletitle{Bayesian nonparametric modeling for causal
  inference}.
\newblock \bibinfo{journal}{{\em Journal of Computational and Graphical
  Statistics\/}} (\bibinfo{year}{2012}).
\newblock


\bibitem[\protect\citeauthoryear{Kang and Schafer}{Kang and Schafer}{2007}]%
        {kang_2007}
\bibfield{author}{\bibinfo{person}{Joseph~DY Kang} {and}
  \bibinfo{person}{Joseph~L Schafer}.} \bibinfo{year}{2007}\natexlab{}.
\newblock \showarticletitle{Demystifying double robustness: A comparison of
  alternative strategies for estimating a population mean from incomplete
  data}.
\newblock \bibinfo{journal}{{\em Statistical science\/}}
  (\bibinfo{year}{2007}), \bibinfo{pages}{523--539}.
\newblock


\bibitem[\protect\citeauthoryear{Kim, Rudin, and Shah}{Kim
  et~al\mbox{.}}{2014}]%
        {kim2014}
\bibfield{author}{\bibinfo{person}{Been Kim}, \bibinfo{person}{Cynthia Rudin},
  {and} \bibinfo{person}{Julie~A Shah}.} \bibinfo{year}{2014}\natexlab{}.
\newblock \showarticletitle{The Bayesian Case Model: A Generative Approach for
  Case-Based Reasoning and Prototype Classification}.
\newblock In \bibinfo{booktitle}{{\em Advances in Neural Information Processing
  Systems 27}}. \bibinfo{pages}{1952--1960}.
\newblock


\bibitem[\protect\citeauthoryear{Kim, Shah, and Doshi-Velez}{Kim
  et~al\mbox{.}}{2015}]%
        {kim2015}
\bibfield{author}{\bibinfo{person}{Been Kim}, \bibinfo{person}{Julie~A Shah},
  {and} \bibinfo{person}{Finale Doshi-Velez}.} \bibinfo{year}{2015}\natexlab{}.
\newblock \showarticletitle{Mind the Gap: A Generative Approach to
  Interpretable Feature Selection and Extraction}.
\newblock In \bibinfo{booktitle}{{\em Advances in Neural Information Processing
  Systems 28}}. \bibinfo{pages}{2260--2268}.
\newblock


\bibitem[\protect\citeauthoryear{Kleinberg, Lakkaraju, Leskovec, Ludwig, and
  Mullainathan}{Kleinberg et~al\mbox{.}}{2017}]%
        {kleinberg_2017}
\bibfield{author}{\bibinfo{person}{Jon Kleinberg}, \bibinfo{person}{Himabindu
  Lakkaraju}, \bibinfo{person}{Jure Leskovec}, \bibinfo{person}{Jens Ludwig},
  {and} \bibinfo{person}{Sendhil Mullainathan}.}
  \bibinfo{year}{2017}\natexlab{}.
\newblock \bibinfo{title}{Human Decisions and Machine Predictions}.
  (\bibinfo{year}{2017}).
\newblock
\showURL{%
\url{http://www.nber.org/papers/w23180}}
\newblock
\shownote{Working paper.}


\bibitem[\protect\citeauthoryear{Kleinberg, Ludwig, Mullainathan, and
  Obermeyer}{Kleinberg et~al\mbox{.}}{2015}]%
        {kleinberg_2015}
\bibfield{author}{\bibinfo{person}{Jon Kleinberg}, \bibinfo{person}{Jens
  Ludwig}, \bibinfo{person}{Sendhil Mullainathan}, {and} \bibinfo{person}{Ziad
  Obermeyer}.} \bibinfo{year}{2015}\natexlab{}.
\newblock \showarticletitle{Prediction policy problems}.
\newblock \bibinfo{journal}{{\em The American Economic Review\/}}
  \bibinfo{volume}{105}, \bibinfo{number}{5} (\bibinfo{year}{2015}).
\newblock


\bibitem[\protect\citeauthoryear{Lakkaraju, Bach, and Leskovec}{Lakkaraju
  et~al\mbox{.}}{2016}]%
        {lakkaraju_2016}
\bibfield{author}{\bibinfo{person}{Himabindu Lakkaraju},
  \bibinfo{person}{Stephen~H Bach}, {and} \bibinfo{person}{Jure Leskovec}.}
  \bibinfo{year}{2016}\natexlab{}.
\newblock \showarticletitle{Interpretable decision sets: A joint framework for
  description and prediction}. In \bibinfo{booktitle}{{\em Proceedings of the
  22nd International Conference on Knowledge Discovery and Data Mining}}.
\newblock


\bibitem[\protect\citeauthoryear{Lakkaraju and Rudin}{Lakkaraju and
  Rudin}{2017}]%
        {lakkaraju_2017}
\bibfield{author}{\bibinfo{person}{Himabindu Lakkaraju} {and}
  \bibinfo{person}{Cynthia Rudin}.} \bibinfo{year}{2017}\natexlab{}.
\newblock \showarticletitle{Learning Cost-Effective Treatment Regimes using
  Markov Decision Processes}.
\newblock \bibinfo{journal}{{\em International Conference on Artificial
  Intelligence and Statistics (AISTATS)\/}} (\bibinfo{year}{2017}).
\newblock


\bibitem[\protect\citeauthoryear{Letham, Rudin, McCormick, and Madigan}{Letham
  et~al\mbox{.}}{2015}]%
        {letham_2015}
\bibfield{author}{\bibinfo{person}{Benjamin Letham}, \bibinfo{person}{Cynthia
  Rudin}, \bibinfo{person}{Tyler~H McCormick}, {and} \bibinfo{person}{David
  Madigan}.} \bibinfo{year}{2015}\natexlab{}.
\newblock \showarticletitle{Interpretable classifiers using rules and Bayesian
  analysis: Building a better stroke prediction model}.
\newblock \bibinfo{journal}{{\em The Annals of Applied Statistics\/}}
  \bibinfo{volume}{9}, \bibinfo{number}{3} (\bibinfo{year}{2015}),
  \bibinfo{pages}{1350--1371}.
\newblock


\bibitem[\protect\citeauthoryear{Marewski and Gigerenzer}{Marewski and
  Gigerenzer}{2012}]%
        {marewski2012heuristic}
\bibfield{author}{\bibinfo{person}{Julian~N Marewski} {and}
  \bibinfo{person}{Gerd Gigerenzer}.} \bibinfo{year}{2012}\natexlab{}.
\newblock \showarticletitle{Heuristic decision making in medicine}.
\newblock \bibinfo{journal}{{\em Dialogues Clin Neurosci\/}}
  \bibinfo{volume}{14}, \bibinfo{number}{1} (\bibinfo{year}{2012}),
  \bibinfo{pages}{77--89}.
\newblock


\bibitem[\protect\citeauthoryear{McDonald}{McDonald}{1996}]%
        {McDonald1996}
\bibfield{author}{\bibinfo{person}{Clement~J. McDonald}.}
  \bibinfo{year}{1996}\natexlab{}.
\newblock \showarticletitle{Medical Heuristics: The Silent Adjudicators of
  Clinical Practice}.
\newblock \bibinfo{journal}{{\em Annals of Internal Medicine\/}}
  \bibinfo{volume}{124}, \bibinfo{number}{1 Part 1} (\bibinfo{year}{1996}),
  \bibinfo{pages}{56--62}.
\newblock


\bibitem[\protect\citeauthoryear{Niculescu-Mizil and Caruana}{Niculescu-Mizil
  and Caruana}{2005}]%
        {niculescu_2005}
\bibfield{author}{\bibinfo{person}{Alexandru Niculescu-Mizil} {and}
  \bibinfo{person}{Rich Caruana}.} \bibinfo{year}{2005}\natexlab{}.
\newblock \showarticletitle{Predicting good probabilities with supervised
  learning}. In \bibinfo{booktitle}{{\em Proceedings of the 22nd international
  conference on Machine learning}}. ACM, \bibinfo{pages}{625--632}.
\newblock


\bibitem[\protect\citeauthoryear{Robins and Rotnitzky}{Robins and
  Rotnitzky}{1995}]%
        {robins_1995}
\bibfield{author}{\bibinfo{person}{James~M Robins} {and}
  \bibinfo{person}{Andrea Rotnitzky}.} \bibinfo{year}{1995}\natexlab{}.
\newblock \showarticletitle{Semiparametric efficiency in multivariate
  regression models with missing data}.
\newblock \bibinfo{journal}{{\it J. Amer. Statist. Assoc.}}
  \bibinfo{volume}{90}, \bibinfo{number}{429} (\bibinfo{year}{1995}),
  \bibinfo{pages}{122--129}.
\newblock


\bibitem[\protect\citeauthoryear{Robins, Rotnitzky, and Zhao}{Robins
  et~al\mbox{.}}{1994}]%
        {robins_1994}
\bibfield{author}{\bibinfo{person}{James~M Robins}, \bibinfo{person}{Andrea
  Rotnitzky}, {and} \bibinfo{person}{Lue~Ping Zhao}.}
  \bibinfo{year}{1994}\natexlab{}.
\newblock \showarticletitle{Estimation of regression coefficients when some
  regressors are not always observed}.
\newblock \bibinfo{journal}{{\em Journal of the American statistical
  Association\/}} \bibinfo{volume}{89}, \bibinfo{number}{427}
  (\bibinfo{year}{1994}), \bibinfo{pages}{846--866}.
\newblock


\bibitem[\protect\citeauthoryear{Rosenbaum and Rubin}{Rosenbaum and
  Rubin}{1983a}]%
        {rosenbaum_1983a}
\bibfield{author}{\bibinfo{person}{Paul~R Rosenbaum} {and}
  \bibinfo{person}{Donald~B Rubin}.} \bibinfo{year}{1983}\natexlab{a}.
\newblock \showarticletitle{Assessing sensitivity to an unobserved binary
  covariate in an observational study with binary outcome}.
\newblock \bibinfo{journal}{{\em Journal of the Royal Statistical Society.
  Series B (Methodological)\/}} (\bibinfo{year}{1983}),
  \bibinfo{pages}{212--218}.
\newblock


\bibitem[\protect\citeauthoryear{Rosenbaum and Rubin}{Rosenbaum and
  Rubin}{1983b}]%
        {rosenbaum_1983b}
\bibfield{author}{\bibinfo{person}{Paul~R Rosenbaum} {and}
  \bibinfo{person}{Donald~B Rubin}.} \bibinfo{year}{1983}\natexlab{b}.
\newblock \showarticletitle{The central role of the propensity score in
  observational studies for causal effects}.
\newblock \bibinfo{journal}{{\em Biometrika\/}} \bibinfo{volume}{70},
  \bibinfo{number}{1} (\bibinfo{year}{1983}), \bibinfo{pages}{41--55}.
\newblock


\bibitem[\protect\citeauthoryear{Rosenbaum and Rubin}{Rosenbaum and
  Rubin}{1984}]%
        {rosenbaum_1984}
\bibfield{author}{\bibinfo{person}{Paul~R Rosenbaum} {and}
  \bibinfo{person}{Donald~B Rubin}.} \bibinfo{year}{1984}\natexlab{}.
\newblock \showarticletitle{Reducing bias in observational studies using
  subclassification on the propensity score}.
\newblock \bibinfo{journal}{{\em Journal of the American statistical
  Association\/}} \bibinfo{volume}{79}, \bibinfo{number}{387}
  (\bibinfo{year}{1984}), \bibinfo{pages}{516--524}.
\newblock


\bibitem[\protect\citeauthoryear{Sull and Eisenhardt}{Sull and
  Eisenhardt}{2015}]%
        {sull_2015}
\bibfield{author}{\bibinfo{person}{Donald Sull} {and}
  \bibinfo{person}{Kathleen~M Eisenhardt}.} \bibinfo{year}{2015}\natexlab{}.
\newblock \bibinfo{booktitle}{{\em {Simple rules: How to thrive in a complex
  world}}}.
\newblock \bibinfo{publisher}{Houghton Mifflin Harcourt}.
\newblock


\bibitem[\protect\citeauthoryear{Tetlock}{Tetlock}{2005}]%
        {tetlock_2005}
\bibfield{author}{\bibinfo{person}{Philip Tetlock}.}
  \bibinfo{year}{2005}\natexlab{}.
\newblock \bibinfo{booktitle}{{\em Expert political judgment: How good is it?
  How can we know?}}
\newblock \bibinfo{publisher}{Princeton University Press}.
\newblock


\bibitem[\protect\citeauthoryear{Ustun and Rudin}{Ustun and Rudin}{2016}]%
        {ustun_2016}
\bibfield{author}{\bibinfo{person}{Berk Ustun} {and} \bibinfo{person}{Cynthia
  Rudin}.} \bibinfo{year}{2016}\natexlab{}.
\newblock \showarticletitle{Supersparse linear integer models for optimized
  medical scoring systems}.
\newblock \bibinfo{journal}{{\em Machine Learning\/}} \bibinfo{volume}{102},
  \bibinfo{number}{3} (\bibinfo{year}{2016}), \bibinfo{pages}{349--391}.
\newblock


\bibitem[\protect\citeauthoryear{Ustun and Rudin}{Ustun and Rudin}{2017}]%
        {ustun_2017}
\bibfield{author}{\bibinfo{person}{Berk Ustun} {and} \bibinfo{person}{Cynthia
  Rudin}.} \bibinfo{year}{2017}\natexlab{}.
\newblock \showarticletitle{Learning Optimized Risk Scores on Large-Scale
  Datasets}.
\newblock \bibinfo{journal}{{\em arXiv preprint arXiv:1610.00168\/}}
  (\bibinfo{year}{2017}).
\newblock


\bibitem[\protect\citeauthoryear{W{\"u}bben and Wangenheim}{W{\"u}bben and
  Wangenheim}{2008}]%
        {wubben2008instant}
\bibfield{author}{\bibinfo{person}{Markus W{\"u}bben} {and}
  \bibinfo{person}{Florian~V Wangenheim}.} \bibinfo{year}{2008}\natexlab{}.
\newblock \showarticletitle{Instant customer base analysis: Managerial
  heuristics often get it right}.
\newblock \bibinfo{journal}{{\em Journal of Marketing\/}} \bibinfo{volume}{72},
  \bibinfo{number}{3} (\bibinfo{year}{2008}), \bibinfo{pages}{82--93}.
\newblock


\bibitem[\protect\citeauthoryear{Zeng, Ustun, and Rudin}{Zeng
  et~al\mbox{.}}{2016}]%
        {zeng2016interpretable}
\bibfield{author}{\bibinfo{person}{Jiaming Zeng}, \bibinfo{person}{Berk Ustun},
  {and} \bibinfo{person}{Cynthia Rudin}.} \bibinfo{year}{2016}\natexlab{}.
\newblock \showarticletitle{Interpretable classification models for recidivism
  prediction}.
\newblock \bibinfo{journal}{{\em Journal of the Royal Statistical Society:
  Series A (Statistics in Society)\/}} (\bibinfo{year}{2016}).
\newblock


\end{thebibliography}

\end{document}